\documentclass[hidelinks,12pt]{article}
\usepackage{amsmath}
\usepackage{amssymb}
\usepackage{amsthm}
\usepackage{relsize}
\usepackage{mathtools}
\usepackage{url}
\usepackage[usenames,svgnames]{xcolor}
\usepackage{enumitem}
\usepackage{comment}
\usepackage{mathdots}
\usepackage{graphicx}
\usepackage{float}
\usepackage[symbol]{footmisc}
\usepackage[pagebackref=false]{hyperref} 
\usepackage[centertableaux]{ytableau} 

\theoremstyle{plain}
\theoremstyle{definition}
\newtheorem{theorem}{Theorem}[section]
\newtheorem{lemma}[theorem]{Lemma}
\newtheorem{definition-theorem}[theorem]{Definition-Theorem}
\newtheorem{definition-proposition}[theorem]{Definition-Proposition}

\newtheorem{proposition}[theorem]{Proposition}
\newtheorem{corollary}[theorem]{Corollary}
\newtheorem{example}{Example}[section]
\newtheorem{examples}{Example}[subsection]

\newtheorem{remark}{Remark}[section]

\newtheorem{definition}{Definition}[section]

\numberwithin{equation}{section} 

\DeclareMathOperator{\End}{End}

\DeclareMathOperator{\diag}{diag}

\def\ra{{\rightarrow}}

\def\det{\mathrm {det}}
\def\Det{\mathrm {Det}}

\def\Pl{\grP\grl}

\def\End{\mathrm {End}}

\def\span{\mathrm {span}}

\def\diag{\mathrm {diag}}

\def\Gr{\mathrm {Gr}}
\def\res{\mathop{\mathrm {res}}\limits}
\def\swedge{\mathsmaller{\mathsmaller{\wedge\,}}}

\def\Gl{\mathrm{Gl}}
\def\Sl{\mathrm{Sl}}

\def\Sp{\mathrm{Sp}}

\def \Cl{\mathrm{Cl}}
\def \symp{\grs\grp}

\def\res{\mathop{\mathrm{res}}\limits}

\def\&{&{\hskip -20pt}}

\DeclarePairedDelimiter{\no}{:}{:}

\def\be{\begin{equation}}
\def\ee{\end{equation}}

\def\bea{\begin{eqnarray}}
\def\eea{\end{eqnarray}}

\def\bt{\begin{theorem}}
\def\et{\end{theorem}}

\def\bex{\begin{example}\small \rm}
\def\eex{\end{example}}

\def\bexs{\begin{examples}\small \rm}
\def\eexs{\end{examples}}

\def\ra{\rightarrow}
\def \ss {\subset}

\def\br{\begin{remark}\small \rm \em}
\def\er{\end{remark}}

\def\BB{{\mathcal B}}

\def\DD {{\mathcal D}}

\def\FF {{\mathcal F}}

\def\HH{{\mathcal H}}
\def\II {{\mathcal I}}

\def\LL{{\mathcal L}}

\def\NN{{\mathcal N}}
\def\OO{{\mathcal O}}

\def\Cb{\mathbf{C}}

\def\Ib{\mathbf{I}}

\def\Nb{\mathbf{N}}

\def\Pb{\mathbf{P}}

\def\Zb{\mathbf{Z}}
\def\ab{\mathbf{a}}
\def\bb{\mathbf{b}}

\def\nb{\mathbf{n}}

\def\sb{\mathbf{s}}
\def\tb{\mathbf{t}}

\def\0b{\boldsymbol{0}}

 \def\grg{\mathfrak{g}}

 \def\grl{\mathfrak{l}}

\def\grP{\mathfrak{P}} \def\grp{\mathfrak{p}}

 \def\grs{\mathfrak{s}}

 \def\grgl{\mathfrak{gl}}

\begin{document}
\baselineskip 16pt

\medskip
\begin{center}
\begin{Large}\fontfamily{cmss}
\fontsize{17pt}{27pt}
\selectfont
	\textbf{Lagrangian Grassmannians, CKP hierarchy and hyperdeterminantal relations}
	\end{Large}
	
\bigskip \bigskip
\begin{large}
S. Arthamonov$^{1, 2}$\footnote[1]{e-mail:artamono@crm.umontreal.ca}
J. Harnad$^{1, 2}$\footnote[2]{e-mail:harnad@crm.umontreal.ca}
and J. Hurtubise$^{1, 3}$\footnote[3]{e-mail:jacques.hurtubise@mcgill.ca}
 \end{large}
 \\
\bigskip
\begin{small}
$^{1}${\em Centre de recherches math\'ematiques, Universit\'e de Montr\'eal, \\C.~P.~6128, succ. centre ville, Montr\'eal, QC H3C 3J7  Canada}\\
$^{2}${\em Department of Mathematics and Statistics, Concordia University\\ 1455 de Maisonneuve Blvd.~W.~Montreal, QC H3G 1M8  Canada}\\
$^{3}${\em Department of Mathematics and Statistics, McGill University, \\ 805 Sherbrooke St.~W.~Montreal, QC  H3A 0B9 Canada }
\end{small}
 \end{center}
\medskip
\begin{abstract}
\smaller{
This work concerns the relation between the geometry of Lagrangian Grassmannians and the CKP integrable hierarchy. 
 The  Lagrange map from the Lagrangian Grassmannian  of maximal isotropic (Lagrangian) subspaces of a finite dimensional 
symplectic vector space $V\oplus V^*$ into the projectivization of the exterior space $\Lambda V$  is defined by restricting 
the Pl\"ucker map  on the full Grassmannian to the Lagrangian sub-Grassmannian  and composing it 
with projection to the subspace of symmetric elements under dualization $V \leftrightarrow V^*$.
In terms of the affine coordinate matrix on the big cell, this reduces to the principal minors map, whose image is cut out by the 
$2 \times 2 \times 2$ quartic {\em hyperdeterminantal} relations. To apply this to the CKP hierarchy, the Lagrangian Grassmannian framework 
is extended to infinite dimensions, with $V\oplus V^*$ replaced by a polarized Hilbert space $ {\mathcal H} ={\mathcal H}_+\oplus {\mathcal H}_-$, 
with symplectic form $\omega$. The image of the Plucker map in the fermionic Fock space ${\mathcal F}= \Lambda^{\infty/2}{\mathcal H}$ 
is identified and the infinite dimensional Lagrangian map is defined. The linear constraints defining 
reduction to the CKP hierarchy are expressed as a fermionic null condition and the infinite analogue 
of the hyperdeterminantal relations is deduced.  A multiparametric family of such relations is shown to be 
satisfied by the evaluation of the $\tau$-function at translates of a point in the 
space of odd flow variables along the cubic lattices generated by power sums in the parameters.}
 \end{abstract}
\break
\section{Integrable hierarchies, Grassmannians, $\tau$-functions}
\label{KP_BKP_tau}

 It is  well known, since the work of Sato  \cite{Sa} and his school and of Segal and Wilson \cite{SW}, that solutions of 
the KP integrable infinite hierarchy of PDE's are determined by abelian group flows on infinite dimensional Grassmann manifolds. 
The associated $\tau$-functions satisfy Hirota bilinear relations \cite{Hir}, which may be interpreted as infinite dimensional 
versions of the Pl\"ucker relations determining the embedding of the Grassmannian in an exterior product space. There are also 
discrete versions of the hierarchy, in which the appropriate recursion relations appear as special Pl\"ucker-type relations \cite{SaSa}, and  
solutions can be related to discrete flows on the Grassmannian \cite{Mi}.

For KP, and finite dimensional reductions, the systems are associated to homogeneous spaces of the $A$-series of Lie groups, or their infinite dimensional limit. 
They have  generalizations, the BKP,  CKP and DKP hierarchies, associated to the $B$, $C$ and $D$-series of  Lie groups, \cite{DJKM1}. There are also discrete versions, expressed as lattice equations \cite{Mi, Hir,  Nim, Ka, Sch, BobSch1, BobSch2, FN}. In the BKP, CKP and DKP cases,  the flows  are on (maximal) isotropic Grassmannians, cut out by quadratic relations on  infinite exterior algebra spaces, interpreted as fermionic Fock spaces, which are
 the infinite analogues of spinor (or Clifford) modules for the $B$ or $D$ series, and of suitably restricted 
linear subspaces for the $C$ series.

The link to isotropic Grassmannians is implicit in previous studies \cite{DJKM1, DJKM2, VdLOS}, of BKP or CKP, but the main focus has
been solutions to the continuous hierarchies, either as vacuum expectation values in
 fermionic Fock space, or equivalent representation theoretic constructions, or application of some form of
the Riemann-Hilbert ``dressing method'' \cite{NMPZ}. Links between the continuous and discrete hierarchies have generally
been based on the use of discrete symmetries, as groups of dressing transformations \cite{Nim}. 

The aim of this paper is to derive the link between continuous and discrete hierarchies through evaluations of the $\tau$-function at infinite lattices of points
embedded within the flow group orbits as Sato did for the case of KP \cite{SaSa}, by suitable interpretations of the addition formulae 
for $KP$ $\tau$-functions. In the case of the CKP hierarchy,  analogous 
formulae are derived using flows restricted to infinite Grassmannians of Lagrangian type. The link with lattice recursion systems involves
quartic relations of the hyperdeterminantal type \cite{Oed, HoSt, Ka}.  More generally, we show how continuous KP $\tau$ functions of the CKP type
provide solutions to the hexahedron relations of Kenyon and Pemantle \cite{KePe1, KePe2}.
 In both finite and infinite dimensions we show, over a generic set,  that the ``short'' Pl\"ucker type relations
  generate the entire  set of  Pl\"ucker relations, and similarly, the lattice recursion relations appearing in the Lagrangian case,
   are in fact, just the ``short'' versions of the full set of quartic relations satisfied by solutions of the discrete CKP hierarchy.  This discretization
therefore provides families of solutions to the hyperdeterminantal relations for every continuous KP $\tau$-function of CKP type.
(See also \cite{BobSch2}, where a link between KP tau functions, hyperdeterminantal relations and the hexahedron reccurences
 is made.)

\subsection{The KP hierarchy, infinite Grassmannians and the Pl\"ucker map}
\label{KP_tau}

In the study of the Kadomtsev-Petviashvili (KP)  hierarchy \cite{Sa, SW, HB}, 
  the  $\tau$-function  $\tau^{KP}_w({\bf t})$ is a key ingredient.   It depends on an infinite sequence  
  of commuting flow variables
\be
{\bf t} = (t_1, t_2, \dots),
\ee
 and is parametrized by elements $w \in \Gr_{\HH_+}(\HH)$ of an infinite Grassmannian \cite{Sa, DJKM1, DJKM2, DJKM3, SW},
consisting of subspaces $w \ss \HH$ of a polarized Hilbert space $\HH=\HH_+ + \HH_-$, commensurate with
the subspace $\HH_+\ss \HH$.  It satisfies the Hirota bilinear residue relation,
\be
\res_{z=\infty }\left(e^{\sum_{i=1}^\infty \delta t_i z^i} \tau^{KP}_w({\bf t} - [z^{-1}])\tau^{KP}_w(\tb +\delta\tb + [z^{-1}])  \right)dz  = 0,
\label{hirota_bilinear_tau_res}
\ee
identically in $\delta \tb$,  where
\be
\delta \tb  :=(\delta t_1, \delta t_2, \dots),  \quad [z^{-1}] := \left({1\over z}, {1\over 2z^2},  \dots,{1\over j z^j}, \dots \right).
\ee

Expanding $\tau^{KP}_w({\bf t})$  in a basis of Schur functions \cite{Mac, Sa}
\be
\tau^{KP}_w({\bf t}) = \sum_{\lambda}\pi_\lambda(w)s_\lambda({\bf t}),
\label{tau_schur_exp}
\ee
with the flow parameters $(t_1, t_2, \dots)$ interpreted as normalized power sums
\be
t_i = \frac{p_i}{i} , \quad p_i:=\sum_{a=1}^\infty x_a^i \quad i=1,2 \dots,
\ee
and  the labels $\lambda$ denoting integer partitions
$\lambda =(\lambda_1 \ge \lambda _2 \ge\cdots \ge \lambda_{\ell(\lambda)} >0, \cdots)$,
the coefficients $\pi_\lambda(w)$ may be  interpreted  as {\em Pl\"ucker coordinates} of the element  $w \in \Gr_{\HH_+}(\HH)$.
These satisfy the  Pl\"ucker relations \cite{GH, HB}, which determine the image of the infinite Grassmannian
 $\Gr_{\HH_+}(\HH)$ under the {\em Pl\"ucker map}:
\bea
 \Pl_{\HH_+, \HH}: \Gr_{\HH_+}(\HH) &\&\ra \Pb(\FF) \cr
\Pl_{\HH_+, \HH}: \span\{w_1, w_2, \dots \} &\&\mapsto \left[w_1 \swedge w_2 \swedge \cdots \right]
=\big[\sum_\lambda \pi_\lambda(w) |\lambda; n\rangle\big] \in \Pb(\FF)
\label{plucker_map_inf}
\eea
embedding $\Gr_{\HH_+}(\HH)$ into the projectivization of the fermionic Fock space $\FF$,
which is the semi-infinite wedge product space
\be
\FF = \Lambda^{\infty/2}(\HH) =\sum_{n\in \Zb}\FF_n.
\label{fock_space}
\ee

Here $\{|\lambda;n\rangle\}$ is the standard basis \cite{Sa, JM1, HB} for the fermionic charge $n$ sector
$\FF_n$ of the Fock space, $\{w_1, w_2, \dots \} $ is an admissible basis \cite{SW} for the
subspace $w \ss \HH$, viewed as an element of the connected component of $\Gr_{\HH_+}(\HH)$,
on which the Fredholm orthogonal projection operator $\Pi_+: w \ra \HH_+$ has index $n$ and $[ |\phi \rangle]\in \Pb(\FF)$
denotes the projective equivalence class of $|\phi \rangle\in \FF$.
As in the finite dimensional case,  the Pl\"ucker coordinates $\{\pi_\lambda(w)\}$  are expressible as determinants of suitably defined
infinite matrices $W_\lambda(w)$, which are maximal minors of the homogeneous coordinate matrix $W(w)$ of the element $w$, relative to an admissible basis \cite{SW, HB}, and may be interpreted  as  holomorphic sections of the (dual) 
determinantal line bundle  $\Det^*\ra \Gr_{\HH_+}(\HH)$.

\subsection{The CKP hierarchy, Lagrangian Grassmannians and the Lagrange map}
\label{CKP_tau_lagrangian}

The CKP hierarchy \cite{DJKM1,  DJKM3, JM1, VdLOS} may similarly be  parametrized by elements $w^0 \in \Gr^\LL_{\HH_+}(\HH, \omega)$ of the
sub-Grassmannian consisting of {\em Lagrangian} (i.e., maximal isotropic)  subspaces of the Hilbert space $\HH$,  
with respect to a complex symplectic product $\omega$ (as defined in Section \ref{symplectic_form_HH}).
It  only involves the odd flow variables
\be
{\bf t}_o = (t_1, t_3, \dots ),
\ee
and the corresponding Baker function satisfies the Hirota bilinear residue equation
{\be
\res_{z=\infty}\left(\Psi_{w^0}(z, {\bf t}_o) \Psi_{w^0}(-z, {\bf t}_o+{\delta\bf t}_o) \right)dz= 0
\label{hirota_bilinear_tau_res_CKP}
\ee
identically in
\be
\delta{\bf t}_o = (\delta t_1, \delta t_3, \dots).
\ee
It may be expressed \cite{CW, CH, KZ}   in terms of a CKP $\tau$-function
  $\tau^{CKP}_{w^0}({\bf t}_o)$ as
\be
\Psi_{w^0}(z, \tb_o) :=
z^{-1/2} \left( \psi_{w^0}(z, \tb_o)\frac{ \partial \psi_{w^0}(z, \tb_o)}{ \partial t_1} \right)^{\frac{1}{2} },
\ee
where
\bea
\psi_{w^0}(z, {\bf t}_o)
&\&:= e^{\tilde{\xi}(z, {\bf t}_o)} \frac{\tau^{CKP}_{w^0} (\tb_o - 2[z^{-1}]_o)}{\tau^{CKP}_{w^0} (\tb_o)},\\
&\& \cr
\tilde{\xi}(z, {\bf t}_o) &\&:= \sum_{j=1}^\infty t_{2j -1} z^{2j -1}, \quad
[z^{-1}]_o := \left(z^{-1}, {1\over 3} z^{-3}, {1\over 5}  z^{-5}, \dots \right).
\eea

The square of $\tau^{CKP}_{w^0}({\bf t}_o)$ is the restriction to vanishing values of the
even KP flow variables $\tb':= (t_1, 0, t_3, 0, \cdots )$, of a KP $\tau$-function $\tau_{w^0}^{KP}(\tb)$
\be
(\tau^{CKP}_{w^0}({\bf t}_o))^2 = \tau^{KP}_{w^0}(\tb')
\ee
satisfying the auxiliary criticality condition \cite{CW, CH, KZ}
\be
\frac{\partial \tau^{KP}(\tb)}{\partial t_2 }\bigg|_{\tb = \tb'} =0
\label{KZ_linear_cond}
\ee
and,  more generally,
\be
\frac{\partial\tau^{KP}(\tb)}{\partial t_{2j} }\bigg|_{\tb = \tb'} =0, \quad  \forall \ j\in \Nb^+.
\label{2j_linear_cond}
\ee
 It follows that we have a Schur function expansion
\be
(\tau^{CKP}_{w^0}({\bf t}_o))^2 = \sum_{\lambda}\pi_\lambda(w^0)s_\lambda(\tb'),
\label{square_schur_expansion}
\ee
in which the Pl\"ucker relations are satisfied by the coefficients $\{\pi_\lambda(w^0)\}$, as well as an infinite set of
linear relations which imply that $w^0 \ss \HH$ is a Lagrangian subspace with respect to the symplectic form $\omega$.

\subsection{Summary of content and results }
\label{summary}

Section \ref{plucker_lagrange} recalls  the setting of finite dimensional Grassmannians, their Pl\"ucker
embedding in a projectivized exterior space and  the Lagrangian  Grassmannian $\Gr^\LL_V(\HH_N, \omega_N)$,
consisting of  subspaces $w^0\ss \HH_N $ of the $2N$-dimensional symplectic vector space
$\HH_N=V\oplus V^*$ that are maximal isotropic with respect to the canonically defined symplectic form $\omega_N$.
In  Section \ref{lagrange_map}, the Lagrange map 
\be
 \LL^N: \Gr^\LL_V(\HH_N, \omega_N) \ra \Pb(\Lambda(V))
 \label{lagrange_map_fin}
 \ee
 is defined, extending the principal  minors map,  defined on the space of $N\times N$ symmetric
  affine coordinate matrices on the big cell, to the entire Lagrangian Grasssmannian.

 The linear coefficients $\LL_J(w^0)$ of the image
\be
\LL^N(w^0) = \big[\sum_J  \LL_J(w^0) e_{-J^c}\big]
\ee
relative to a basis $\{e_{-J^c}\}$ for $\Lambda(V)$
labelled by ordered subsets $J \ss \{1, \dots, N\}$ of integers  (where $J^c$ is the complement of $J $  )
 coincide with the Plücker coordinates $\pi_{\lambda}(w^0)$ corresponding to  {\em symmetric} partitions $\lambda=\lambda^T$.
    However, the map $\LL$ is not one-to-one (cf. \cite{VGM}). As explained in Section \ref{lagrange_hyperdet_relations_N} its fibres  are
    the orbits of the group $(\Zb_2)^N$ of reflections within the symplectic $2$-planes corresponding to a canonical basis   and,
  generically, are of cardinality $2^{N-1}$.

    For Lagrangian subspaces $w^0 \in \Gr^\LL_V(\HH_N, \omega_N)$ in the  {\em big cell},  the $\LL_J$'s are the 
    principal minor determinants of the $N\times N$ symmetric affine coordinate matrix $A(w^0)$.
  As shown in \cite{HoSt, Oed}, these satisfy  the set of quartic relations (\ref{hyper_det_rels}),
  the ``core'' {\em hyperderminantal relations}, whose orbit under the symplectic subgroup
 \be
 G_N := \left(\Sl(2)\right)^N \rtimes S_N \ss \Sp(\HH_N, \omega_N),
 \label{G_N_def}
 \ee
  cuts out the image of the Lagrange map.
 Combining the quadratic Pl\"ucker relations  with the linear conditions on the Pl\"ucker coordinates
 which assure that the element $w^0$ is in the Lagrangian Grassmannian $ \Gr^\LL_V(\HH_N, \omega_N) $,
 a new proof of these relations, valid on a Zariski open subset, is provided  in  Sections \ref{geometry_plucker_restriction}
 and \ref{hyperdet_Gr_L_3_6} (Propositions \ref{prop:plucker-hyper}, \ref{prop:N_3_hyper_red}).
 It is also shown how a more general set of relations, the {\em hexahedron reccurrence}
 relations, introduced in \cite{KePe1, KePe2} in the study of  double dimer coverings and rhombus tilings, follow from
 the Pl\"ucker relations and isotropy conditions for Lagrangian Grassmannians.

The realization of the KP hierarchy in terms of isospectral flows of formal pseudo-differential operators  is
recalled in Section \ref{KP_Baker_CKP}, together with its reduction to the CKP case.
The Grassmannian interpretation of this reduction consists of restricting the KP flows on  the infinite
Grassmannian $Gr_{\HH_+} (\HH)$ of subspaces of  the underlying  polarized
Hilbert space $\HH= \HH_- \oplus \HH_+$  of the KP hierarchy, commensurable with $\HH_+$,
 to the subgroup of flows in the odd flow parameters only, acting on the Lagrangian sub-Grassmannian
$\Gr^\LL_{\HH_+} (\HH, \omega) \ss \Gr_{\HH_+} (\HH)$ of isotropic subspaces with respect to a suitably defined
symplectic form $\omega$ on $\HH$. The fermionic representation of the KP $\tau$-function as a vacuum expectation value (VEV)
on the associated fermionic Fock space $\FF =\Lambda^{\infty/2}(\HH)$ is recalled in Section \ref{fermionic_KP-tau}.

The symplectic form  $\omega$  on  $\HH$  is introduced in Section \ref{symplectic_form_HH} and used to define the
 infinite symplectic group $\Sp(\HH, \omega)$ action on $\HH$ and on $ \FF$.
In Section \ref{CKP_reduction}, the reduction conditions from the KP to the CKP 
hierarchy are expressed as  fermionic null conditions equivalent to the Lagrangian condition.
 Using the bosonization map and the Murnaghan-Nakayama rule, this is shown to imply an infinite
 set of linear vanishing conditions (Proposition \ref{prop:null_vanishing_conds}) satisfied by the Pl\"ucker coordinates.

The infinite dimensional analog of the Lagrange map
\be
\LL: \Gr^\LL_{\HH+}(\HH, \omega) \ra \Pb(\FF^S)
\ee
 is introduced in Section \ref{lagrange_map_inf},  mapping the Lagrangian Grassmannian $\Gr^\LL_{\HH_+} (\HH, \omega)$
   to the projectivization of the  subspace $\FF^S = \Lambda^{\infty/2}\HH_+ \ss \FF$ spanned by basis elements corresponding to symmetric partitions.
Combining  the Pl\"ucker relations  with the Lagrangian condition, it is shown in Section \ref{inf_hyperdet_rels} (Proposition \ref{prop:hypedet_inf_red}),
that the symmetric partition Pl\"ucker coordinates of an element  $w^0\in \Gr^\LL_{\HH+}(\HH, \omega) $ corresponding
to a CKP type $\tau$-function satisfy the hyperdeterminantal relations.
Finally, in Section \ref{param_hyperdet_tau}  it is shown (Proposition \ref{hyperdet_3_param_family}}
and  Corollary \ref{tau_functional_N_param_hyperdet}), as a consequence of the addition formulae for KP $\tau$-functions 
(generalized Fay identities), that an $N$-parameter family of hyperdeterminantal  relations is satisfied by the $\tau$-function,  
evaluated at the translates of a point in the space of odd flow variables by cubic lattices generated by power sums in the parameters.

\makeatletter
\@addtoreset{equation}{section}
\makeatother
\renewcommand{\theequation}{\thesection.\arabic{equation}}

\makeatletter
\@addtoreset{equation}{subsection}
\makeatother
\renewcommand{\theequation}{\thesubsection.\arabic{equation}}

\section{Pl\"ucker map, Clifford algebra and Lagrange map in finite dimensions}
\label{plucker_lagrange}

\subsection{The Pl\"ucker map and Pl\"ucker relations}
\label{plucker_map_relations}
  The Pl\"ucker map  \cite{GH}
\bea
\Pl^n_k: \Gr_k(\Cb^n)&\& \ra  \Pb(\Lambda^k(\Cb^n)) \cr
\Pl^n_k: w &\& \mapsto  [W_1\wedge \cdots \wedge W_k]
\eea
(where $[\phi]$ denotes the projective equivalence class of $\phi \in \Lambda^k(\Cb^n)$)
defines an embedding of the Grassmannian $\Gr_k(\Cb^n)$ of $k$-planes 
$
w=\span\{W_1, \dots, W_k\} \ss \Cb^n
$
in  the projectivization $\Pb(\Lambda^k(\Cb^n))$ of the exterior space
$\Lambda^k(\Cb^n)$. It  is equivariant with respect to the natural action of the
general linear group $\Gl(n, \Cb)$ on $\Gr_k (\Cb^n)$ and on $\Pb(\Lambda^k(\Cb^n))$.
The image $\Pl^n_k(\Gr_k(\Cb^n)) \in \Pb(\Lambda^k(\Cb^n))$  is the intersection of a number of quadrics, the {\em Pl\"ucker quadrics},
thereby realizing $\Gr_k(\Cb^n)$ as a projective variety. The Pl\"ucker coordinates $\pi_\lambda(w)$ are the (projectivized)
linear coordinates of the image $\grP^n_k(w)$ in the standard basis $\{f_L\}_{L=(L_1, \dots, L_k))}$
for the exterior space $\Lambda^k(\Cb^n)$,  defined by
\be
f_L := f_{L_1} \swedge \cdots \swedge f_{L_k},
\ee
where the multi-index
\be
L:=(L_1, \dots, L_k)  \ss \{1, \dots ,n\}
\ee
is a $k$-element subset of $\{1, \dots ,n\}$, written in increasing order and $\{f_1, \dots, f_n\}$
is the standard basis for  $\Cb^n$.
Thus
\be
\grP^n_k(w) =\left[ \sum_\lambda \pi_\lambda(w) f_L\right],
\ee
where the partition $\lambda = (\lambda_1 \ge \cdots \ge \lambda_k \ge 0)$
associated to $L$,  which labels the Pl\"ucker coordinate $\pi_\lambda(w)$, is given by
\be
\lambda_i  = L_{k-i+1} +i -k -1, \quad i=1, \dots k
\ee
and its Young diagram fits into a $k \times (n-k)$ rectangle.

Equivalently, let $W$ be the $n\times k$ homogeneous coordinate matrix of the element $w$,
whose columns are the basis vectors $(W_1, \dots, W_k)$, viewed as column vectors,
and let $W_\lambda$ be the $k \times k$ matrix whose $i$th row is the $L_i$th row of $W$.
Then
\be
\pi_\lambda(w) =  \det(W_\lambda).
\ee
The labelling by partitions $\lambda$ or by $k$ element subsets $L\ss \{1, \dots, n\}$
is equivalent, but it is sometimes more convenient to use the  multi-index $L$,
in which case we write
\be
\tilde{\pi}_L(w) := \pi_\lambda(w).
\ee

 The fact that $\grP^n_k(w) \in \Pb(\Lambda^k(\Cb^n))$ is (the projectivization of) a completely decomposable element
of $\Lambda^k(\Cb^n)$ uniquely characterizes the image  of the Pl\"ucker map.  It is equivalent to $\Pl^n_k(w)$  satisfying
the quadratic {\em Pl\"ucker relations}, which are obtained by contracting it, as a pojectivized $k$-vector,
 with the various possible basis elements in $\Lambda^{k-1}( \Cb^{n*})$, to obtain elements of $w$,  and noting that,
 due to the decomposability of $\grP^n_k(w)$, their exterior products with $\grP^n_k(w)$ must vanish.
 The vanishing of the components of the resulting  elements of $\Lambda^{k+1}( \Cb^{n})$,
 expressed in terms of Pl\"ucker coordinates, define the Pl\"ucker relations.

To express these concisely \cite{GH},  let $(I, J)$ be a pair of ordered subsets of $\{1, \dots ,n\}$
with cardinalities $k-1$ and $k+1$, respectively:
\bea
I &\&= (I_1,I_2,\dots, I_{k-1}), \quad1 \leq I_{1} < I_{2} < \dots < I_{k-1} \leq n, \cr
J &\&= (J_1,J_2,\dots, J_{k+1}),\quad 1 \leq J_{1} < J_2 < \dots < K_{k+1}\leq n.
\eea
For any ordered subset
\be
L=(L_1, \dots L_r),  \quad 1 \le L_1 < \cdots < L_r \le n
\ee
of cardinality $r$, and any $j\in \{1, \dots ,n\}$, $j \notin L$, denote by $L(j)$
the ordered set with elements $(L_1, \dots L_r, j)$ and
\be
 (L_1, \dots, \widehat{L}_m, \dots L_{r }), \quad m=1, \dots r
\ee
the subset $L\backslash \{L_m\}$ with $L_m$ removed.
The Pl\"ucker relations are then
\be
\sum_{m=1}^{k+1}(-1)^{m}\tilde{\pi}_{(I_1, I_2 \dots \dots I_{k-1}  J_m)} \tilde{\pi}_{(J_1 J_2\dots, \widehat{J}_m, \dots J_{k+1} )}=0,
\label{plucker_rel_kn}
\ee
where the indexing has been extended to all multi-index distinct sequences, such that  Pl\"ucker coordinates whose indices differ
by a permutation from the increasingly ordered sequence are understood to equal the ordered one times the sign of the permutation.

The relations (\ref{plucker_rel_kn}) are not independent, of course. Generically, a much smaller subset,
 known as the {\em short } Pl\"ucker relations,  in which the intersection $I\cap J$ is of cardinality $k-2$,
 suffices to generate them all. If we choose the first $k-2$ of the indices to coincide
\be
I':= (I_1=J_1, \cdots  , I_{k-2} = J_{k-2}),
\ee
there are only three possible distinct terms in the sum (\ref{plucker_rel_kn}).
Letting
\be
I_{k-1} := i \quad J_{k-1} := j_1 , \quad J_{k} :=j_2, \quad J_{k+1} := j_3,
\ee
these are
\be
\tilde{\pi}_{(I, j_1)} \tilde{\pi}_{(I', j_2, j_3 )}
+\tilde{\pi}_{(I, j_3)} \tilde{\pi}_{(I', j_1, j_2)}
+\tilde{\pi}_{(I, j_2)} \tilde{\pi}_{(I', j_3, j_1)} =0.
\label{short_plucker}
\ee

As shown in \cite{HB}, App. D  (cf.  also \cite{KPRS}),  on a Zariski open set within $\Pb(\Lambda^k(\Cb^n)$, these
short Pl\"ucker relations are sufficient to imply the full set. This follows inductively from the Desnanot-Jacobi identity,
and the generalized Giambelli identity, which expresses all Pl\"ucker coordinates as minor
determinants of the matrix of hook partition Pl\"ucker coordinates.  Another proof of this fact, formulated more
geometrically, is provided in Section \ref{geometry_plucker_restriction}.


\subsection{Pl\"ucker map for $\Gr_V(\HH_N)$ and the Clifford algebra}
\label{plucker_map_clifford_alg}

Let $V$ be a complex vector space of dimension $N$, $V^*$ its dual space, and denote by
 \be
 \HH_N := V \oplus V^*
 \label{HH_N_def}
 \ee
the direct sum of the two.
The Grassmannian $\Gr_V(\HH_N)$ of $N$-planes in $\HH_N$ is the orbit of  $V\ss \HH_N$ under the
action of the general linear group $\Gl(\HH_N)$.
The {\em Pl\"ucker map}
\be
\Pl_V:\Gr_V(\HH_N) \ra \Pb(\Lambda^N(\HH_N))
\label{plucker_map}
\ee
for this case is the $\Gl(\HH_N)$ equivariant embedding  of $\Gr_V(\HH_N)$
in the projectivization \hbox{$ \Pb(\Lambda^N(\HH_N))$ } of the exterior space $\Lambda^N(\HH_N)$
defined   by:
\be
\Pl_{V} : w \mapsto [w_1 \swedge \cdots \swedge w_N] \in \Pb(\Lambda^N(\HH_N)),
\label{plucker_map_def}
\ee
where  $\{w_1, \dots, w_N\}$ is a basis for the subspace $w\in \Gr_V(V\oplus  V^*)$.
 Its image is cut out by the intersection of the {\em Pl\"ucker quadrics}
  (\ref{plucker_rel_kn}), for $k=N$, $n=2N$.

 To anticipate the notational conventions used in the next section, we index the basis
 for $V$  and $V^*$ henceforth as $\{e_{-j}\}_{j=1, \dots, N}$ and $\{e_j\}_{j=0, \dots, {N-1}}$ respectively,
 with dualization pairing
    \be
  e_i(e_{-j})= (-1)^i\delta_{i+1, j}.
  \label{dual_basis_pairing}
  \ee
Ordering the basis  for $\HH_N$  as $(e_{-N}, \dots, e_{-1}, e_0, \dots , e_{N-1})$,
define the corresponding basis elements $\{|\lambda\rangle\}$ for $\Lambda^N(\HH_N)$ by
\be
|\lambda\rangle := e_{l_1} \swedge \cdots \swedge e_{l_N},
\label{basis_finite}
\ee
where $\lambda$ is any partition whose Young diagram fits in the $N\times N$ square diagram,
 and
\be
l_j := \lambda_j -j, \quad 1\le j \le N
\label{particle_coords}
\ee
 are the {\em particle positions}   associated to the partition (see \cite{HB}, Chapt. 5, Sec. 5.1)
 \be
 \lambda = (\lambda_1, \dots, \lambda_{\ell(\lambda)}, 0, \dots ).
 \ee
 Thus $l_1 > \cdots > l_N$ is a strictly decreasing sequence of $N$ integers between $N-1$ and $-N$.
 The ``vacuum'' (or highest weight) vector is defined as
 \be
 |0 \rangle := |\emptyset\rangle = e_{-1} \swedge \cdots \swedge e_{-N},
 \ee
 and its projectivization is the image $\Pl_V(V)$ of $V$ under the Pl\"ucker map.
A (complex) scalar product on $\Lambda^N (\HH_N)$ is defined, in bra/ket notation,
 by requiring the  $\{|\lambda\rangle\}$ basis to be orthonormal
\be
\langle\lambda | \mu\rangle  =\delta_{\lambda\mu}.
\ee

 Following Cartan \cite{Ca, Ch}, define the natural complex scalar product $Q$ on $\HH_N\oplus \HH_N^*$ by
\be
Q((X, \xi), (Y, \eta)) = \eta(X)  + \xi(Y), \quad X, Y \in \HH_N, \quad \xi, \eta \in \HH_N^*,
\ee
and let $\Cl\left(\HH_N \oplus \HH_N^*, Q)\right)$ denote the corresponding
Clifford algebra  on $\HH_N \oplus  \HH_N^*$. The standard irreducible representation
\bea
\Gamma: \Cl\left(\HH_N \oplus \HH_N^*, Q)\right) &\&\ra  \End(\Lambda(\HH_N)),  \cr
\Gamma: \sigma &\&\mapsto \Gamma_\sigma
\eea
is generated by the linear elements, which are represented by exterior and interior multiplication:
\bea
\Gamma_{v + \mu}:= v  \swedge &\&+ i_{\mu},  \  \in \End(\Lambda(\HH_N)),  \\
v \in \HH_N, &\& \quad \mu\in \HH_N^*.
\nonumber
\label{Gamma_w}
\eea
The representations of the basis elements, denoted
\be
\psi_i := \Gamma_{e_i} = e_i \swedge \quad \psi^\dag_i:= \Gamma_{e^*_i} = i_{e^*_i}, \quad i= -N, \dots, N-1,
\label{psi_psi_dag_def}
\ee
 are viewed as finite dimensional fermionic creation and annihilation operators,
which satisfy the anticommutation relations
\be
[\psi_i, \psi_j]_+=0,\quad [\psi^\dag_i, \psi^\dag_j]_+=0, \quad [\psi_i, \psi^\dag_j]_+= \delta_{ij}
\ee
as well as the vacuum annihilation conditions
\be
\psi_{-i} | 0\rangle = \psi^\dag_{i-1}|0\rangle = 0, \quad i= 1, \dots, N.
\ee

\subsection{Pl\"ucker coordinates on $\Gr_V(\HH_N)$}
\label{Grassmannian_Gr_V_WW_plucker}

   For consistency with standard notations \cite{JM1, Sa, HB} used in infinite dimensions
   (see Section \ref{CKP_hierarchy}),  we index our bases as $\{e_{-N}, \dots, e_{-1} \}$ and $\{e_0, \dots, e_{N-1}\}$
to identify $V$ and $V^*$ with $\Cb^N$ and $\Cb^{N*}$, respectively, with the dualization pairing
\be
e_i(e_{-j}) = (-1)^i \delta_{i, j-1},\quad i = 0. \dots, N-1, \quad j =1, \dots, N.
\ee
The dual basis $\{e^*_{-N}, \dots, e^*_{-1}, e_0^*, \cdots, e_{N-1}^*\}$  is thus given by
\be
e^*_i = (-1)^{i+1} e_{-i-1} \quad i = -N, \dots N -1.
\label{dual_basis_isom}
\ee

Let $w\in \Gr_V(\HH_N)$ be an element of the Grassmannian of $N$-dimensional subspaces
of  $\HH_N$, and let $W$ denote the  $2N \times N$ dimensional homogeneous coordinate
matrix  whose columns $\{W_i \in \HH_N\}_{i=1, \dots, N}$ are a  basis for
$w$ expressed relative to $\{e_{-N}, \dots, e_{-1}, e_0, \dots, e_{N-1}\}$.
 The Pl\"ucker coordinates $\{\pi_\lambda(w)\}$
are thus labelled by partitions $\lambda$ whose Young diagrams fit into
the $N\times N$ square $(N)^N$. Recall that any partition
 $\lambda =(\lambda_1 \ge \cdots \lambda_{\ell(\lambda)} \ge 0, \cdots)$ of length $\ell(\lambda) \le N$ may equivalently by labelled
 by its Frobenius indices \cite{Mac}
\be
\lambda = \lambda(\ab | \bb) , \quad (\ab | \bb):=(a_1, \dots, a_r | b_1, \dots , b_r),
\ee
where the Frobenius rank  $r$, with $0\le r \le N$, is the number of diagonal terms in the
Young diagram of $\lambda$, and
\be
\ab = (a_1> \cdots > a_r \ge 0),\quad \bb = (b_1> \cdots > b_r \ge 0)
\ee
are two strictly decreasing sequences of nonnegative integers that represent the ``arm'' \
and ``leg'' lengths' in the Young diagram (i.e., the number of squares to the right of
and below the $r$ diagonal elements, respectively).

To each partition $\lambda \ss (N)^N$, we associate the $N\times N$ submatrix $W_\lambda$
of the homogeneous coordinate matrix whose rows consist of the rows of $W$ in positions
$L_1 < L_2 \cdots < L_N$, where
\be
L_i := l_{N-i+1}, \quad i=1, \dots, N.
\ee
The Pl\"ucker coodinates $\pi_\lambda(w)$ are given, up to projective equivalence, by
the determinants
\be
\pi_\lambda(w) = \det(W_\lambda) =\tilde{\pi}_L.
\ee

For $0 \le  r \le N$, let
\be
\{I:=(I_1, \dots, I_{r}) \ss (1,\dots, N)\}, \quad \{J:=(J_1, \dots, J_{r}) \ss (1,\dots, N)\},
\ee
be a pair of (increasingly) ordered subsets of $(1, \dots, N)$,  with cardinalities $| I |= | J|=r$,
Define a basis  $\{e_{(I,J)}\}$ for $\Lambda^N(\HH_N)$, labelled by such pairs $(I,J)$,
as
\be
e_{(I,J)}:= e_{I-1} \swedge e_{-J^c}, \quad e_{I-1} \in \Lambda^{r}(V^*),
\quad e_{-J^c} \in \Lambda^{N-r}(V),
\quad r=0, \dots, N,
\label{eq:eIJNotationFinite}
\ee
where
\be
J^c = (J^c_1 < \cdots < J^c_{N-r})
\ee
 is the (increasingly ordered) complement of $J \ss (1, \dots,  N)$,  and
\be
e_{I-1} := e_{I_r-1} \swedge \cdots \swedge e_{I_{1}-1} \in \Lambda^r (V^*),
\quad e_{-J^c}:= e_{-J^c_1} \swedge \cdots \swedge e_{-J^c_{N-r}} \in \Lambda^{N-r} (V)
\ee
are the corresponding standard basis elements for $\Lambda^r (V^*)$ and $\Lambda^{N-r} (V)$, respectively.
(Note that the ordering of successive factors in the wedge products is chosen in both cases to be decreasing from left to right.)
 To each partition $\lambda \ss (N)^N$ of Frobenius rank $r$, we associate the unique pair $(I,J)$ such that
\be
 I_i = \lambda_{r-i+1} - r+i, \quad  1 \le i \le r , \quad  J^c_i = r+ i -\lambda_{r+i} ,  \quad 1 \le i \le N-r. \quad
 \label{particle_coords_IJ}
\ee

\begin{lemma}
The relation between the Frobenius indices $(\ab | \bb)$ and the pairs $(I, J)$ is given by:
\be
a_i = I_{r-i +1} -1, \quad b_i =  J_{r-i+1} -1, \quad i = 1, \dots , r, \quad r=0, \dots, N-1.
\label{ab_IJ}
\ee
and the basis elements are related by
\be
e_{(I,J)} = |\lambda \rangle = (-1)^{\sum_{i=1}^r b_i}\prod_{i=1}^r \psi_{a_i}\psi^\dag_{-b_i-1} | 0 \rangle.
\ee
\end{lemma}
\begin{proof}
This follows by direct application of eqs.~(\ref{basis_finite}), (\ref{particle_coords}), (\ref{particle_coords_IJ})
and the definition of the Frobenius indices $(a_1, \dots, a_r | b_1, \dots, b_r)$.
\end{proof}

The pairs $\{(I, J)\}$ thus provide an equivalent labelling of the partitions $\{\lambda \ss (N)^N\}$,
which we denote
\be
\lambda(I, J) := (\ab | \bb) = \lambda^T(J,I).
\ee
We may replace the pair $(I,J)$ by the Frobenius indices  $(\ab | \bb)$, and label the basis elements
equivalently as
\be
|\lambda(I,J)\rangle =|(\ab | \bb) \rangle = e_{(I, J)}.
\ee
The element obtained by interchanging $(I, J)$   just corresponds to the transposed partition
\be
e_{(J, I)} = |(\bb | \ab)\rangle = |\lambda^T \rangle.
\ee
The Pl\"ucker map is thus
\bea
\Pl_V: \Gr_V(\HH_N) &\&  \ra \Pb(\Lambda^N(\HH_N)) \cr
\Pl_V:w &\& \mapsto [W_1 \wedge \cdots \wedge W_N]
 = [\sum_{r=0}^N\sum_{(I, J)} \pi_{\lambda(I,J)}(w) e_{(I,J)}].
\eea
A  {\em symmetric} partition is one that equals its transpose $\lambda = \lambda^T$,  so that $I = J$,  $\ab = \bb$,
with $(I,J)$ related to $(\ab, \bb)$ by (\ref{ab_IJ}).

\begin{remark} Note that, following standard usage, $|\lambda|$ denotes the {\em weight} of an integer partition $\lambda$
(i.e., the sum of its parts), while $\ell(\lambda)$ denotes its {\em length} (i.e., the number of nonzero parts).
For multi-indices $K=(K_1, \dots, K_m)$, however,  $|K | = m$ denotes the {\em cardinality}. There should be no confusion,
since we consistently use lower case Greek letters  $\lambda, \mu, \dots$ for partitions,  upper
case Roman letters $K=(K_1, \dots, K_m)$ for increasingly ordered multi-indices and lower case 
Roman letters $(l_1, l_2, \dots)$ for (finite or infinite) strictly decreasing sequences of integers.
\end{remark}

\subsection{Symplectic form and Lagrange map }
\label{lagrange_map}

Define the symplectic form $\omega_N \in \Lambda^{2}(\HH^*_N)$ on $\HH_N=V\oplus  V^*$  as
 \be
 \omega_N(u+\alpha, v+\beta) =  \alpha(v) - \beta(u)  \quad u, v \in V, \ \alpha, \beta \in V^*.
 \label{eq:omega_N_def}
 \ee
 In terms of the basis elements this is
\be
\omega_N:= \sum_{i=1}^N (-1)^{i +1}e^*_{-i}\swedge e^*_{i-1}   = \sum_{i=1}^N (-1)^i e_{-i} \swedge e_{i-1}.
\label{eq:omega_N_basis}
\ee
where we identify $\HH_N \sim \HH^*_N$  via the isomorphism (\ref{dual_basis_isom}).
The symplectic group $\Sp(\HH_N, \omega_N)$ is thus the subgroup of $\Gl(\HH_N)$ that preserves $\omega_N$:
\be
\Sp(\HH_N, \omega_N) = \{ g \in \Gl(\HH_N)\  |\  \omega_N(g X, g Y) = \omega_N(X, Y), \quad \forall \ X, Y \in \HH_N\}.
\ee
The Lagrangian Grassmannian $\Gr^\LL_V(\HH_N, \omega_N) \ss \Gr_V(\HH_N)$  consists of  those elements
$\{w^0\in \Gr_V(\HH_N) \}$ on which the restriction of $\omega_N$ is totally null:
 \be
 \omega_N \big|_{w^0} = 0.
 \ee
 The ``big cell'' in  $\Gr^\LL_V(\HH_N,\omega_N)$ consists of  elements $w^0\in \Gr_V(\HH_N) $ of the form
 \be
w^0 :=  \span\{e_{-i }+ \sum_{j=1}^N A_{ij}(w^0)(-1)^{j-1} e_{j-1}\}_{i=1, \dots, N},
\label{affine_matrix}
 \ee
 where $A(w^0) = A^T(w^0) $ is a symmetric $N \times N$ matrix, whose entries are
 the {\em affine coordinates} of $w^0$.

 The exterior space $\Lambda(V) \ss \Lambda^N(\HH_N)$ may be identified with the subspace
 \be
 \Lambda^S(\HH_N) \ss\Lambda^N(\HH_N)
 \ee
  spanned by basis elements $\{e_{(J,J)}=|\lambda\rangle\}$ corresponding to symmetric partitions 
  $\lambda = \lambda^T$ via the injection map
\bea
\iota_{\Lambda(V)}: \Lambda(V) &\&\ra \Lambda^S(\HH_N) \ss  \Lambda^N(\HH_N) \cr
&\& \cr
\iota_{\Lambda (V)}: \sum_{J }\LL_J \ e_{-J^c}&\&\mapsto \sum_{J } \LL_J  e_{(J,J)},
\label{LambdaV_inject}
\eea
 where  the sum is over all increasingly ordered multi-indices $J =(J_1, \cdots J_r) \ss (1, \dots, N)$  of
cardinality $0\le r \le N$.

Viewing $\{e_{(I,J)}\}_{I, J\subseteq (1, \dots, N), \atop | I |+ | J |=N}$ as an orthonormal basis
for $\Lambda^N(\HH_N)$ and identifying $\Lambda(V)$ with its image  under the injection map $\iota_{\Lambda(V)}$,
 we have the orthogonal projection
\bea
{\Pr}_{\Lambda (V)}: \Lambda^N(\HH_N) &\&\ra \Lambda(V) \cr
&\& \cr
{\Pr}_{\Lambda (V)}: \sum_ {I, J \subseteq (1, \dots, N)\atop | I |=|J^c |} \pi_{(I,J)} e_{(I,J)} &\&
\mapsto \sum_ {J \subseteq (1, \dots, N)} \pi_{(J,J)} e_{-J^c}
\eea

  \begin{definition}
 The {\em Lagrange map}
 \label{lagrange_map_def_finite}
 \be
 \LL^N: \Gr^\LL_V(\HH_N,\omega_N)) \ra \Pb(\Lambda(V))
 \label{Lagrange_map_finite}
 \ee
 is defined to be the composition of the restriction of the Pl\"ucker map $\Pl_V |_{ \Gr^\LL_V(\HH_N,\omega_N))}$ to
 $\Gr^\LL_V(\HH_N,\omega_N)$ with the projection $\Pr_{\Lambda(V)}$:
 \be
   \LL^N  :=   {\Pr}_{\Lambda(V)} \circ \Pl_V |_{\Gr^\LL_V(\HH_N,\omega_N)}.
   \label{lagrange_map_finite_def}
 \ee
 \end{definition}
 It is therefore expressed in terms of the basis  as
\be
 \LL^N(w^0) = \big[ \sum_{J \subseteq (1, \dots, N)}  \LL_J(w^0) e_{-J^c}\big],
  \label{lagrange_map_finite_basis}
 \ee
where
\be
\LL_J(w^0) := \pi_{\lambda(J, J)}(w^0)
\label{lagrange_coord}
\ee
will be referred to as the {\em Lagrange coefficients}.
   It follows from the generalized Giambelli identity (\cite{HB}, Appendix C) that on the big cell of $\Gr^\LL_V(\HH_N,\omega_N)$,
the Pl\"ucker coodinates are, up to projective equivalence, the minor determinants of the affine coordinate matrix $A(w^0)$.
In particular, for $w^0$ in the  big cell, the $\LL_J(w^0)$'s are, within projective equivalence,
 the determinants of the principal submatrices $A_J(w^0)$ with rows and columns  in $J$
   \be
   \LL_J(w^0) = \det\left(A_J(w^0)\right).
   \label{lagrange_coord_affine}
   \ee

Thus,  for $w^0$ in the big cell, $\LL$  maps $w^0$ to an element of $\Pb(\Lambda(V))$ 
which  may  be expressed as 
   \be
   \LL^N(w^0) =[\sum_ {I \subseteq (1, \dots, N)}\det(A_J(w^0)) e_{-J^c} ].
   \label{lagrange_affine_coord}
   \ee
in the standard basis $\{e_{-J^c}\}_{J\ss (1, \dots N)}$.

\subsection{Decomposition of  $\Lambda(\HH_N)$ into irreducible representations of $\Sp(\HH_N, \omega_N)$ }
 \label{decomp_Sp_2N_irreps}

Viewed as an endomorphism of $\Lambda(\HH_N)$, the inner product with the  symplectic form  will be denoted
\bea
\hat{ \omega}^\dag_N:\Lambda(\HH_N) &\& \ra \Lambda(\HH_N) \cr
  \hat{\omega}^\dag_N: \mu &\& \mapsto i_{\omega_N}\mu,
 \label{eq:OperatorOmegaNDag}
\eea
and the (dual) exterior product as
\nopagebreak{
\bea
   \hat{\omega}_N:\Lambda(\HH_N)&\& \ra \Lambda(\HH_N) \cr
   \hat{\omega}_N :  \mu &\& \mapsto \omega_N \swedge \mu.
   \label{eq:OperatorOmegaN}
\eea
In terms of fermionic creation and annihilation operators, these can be written as
\bea
\hat{\omega}^\dag_N &\&= \sum_{i=0}^{N-1}(-1)^i\psi^\dag_{-i-1}\psi^\dag_i ,
\label{eq:OperatorOmega_dag_N}
\\
\hat{\omega}_N &\&= -\sum_{i=0}^{N-1}(-1)^i\psi_{-i-1}\psi_{i} .
\label{eq:OperatorOmega_N}
\eea

\begin{definition}
For every pair of integers $j, k\in\Nb$ satisfying 
\be
 0\leq j\leq N, \quad 2j\leq k\leq N+j,
\ee
 define the subspace $P^{k}_{k-2j} \ss \Lambda^k(\HH_N)$  as
\be
P^k_{k-2j}= (\hat\omega_N)^j\left(\ker(\hat\omega_N^\dag )\big\vert_{\Lambda^{k-2j}(\HH_N)}\right)\ss \Lambda^k(\HH_N).
\label{Pk_k-2j_def}
\ee
In particular $P^N_N \ss \Lambda^N(\HH_N)$ is defined by either of the equivalent linear relations
\be
\phi \in P^N_N \ \text{if and only if } \  \hat{\omega}^\dag_N(\phi) =0,  \quad
\phi  \in P^N_N \ \text{if and only if }  \  \hat{\omega}_N(\mu)=0.
\label{linear_lagrange_rels}
\ee
\end{definition}
The following is a standard result in the representation theory of $\Sp(\HH_N, \omega_N)$ \cite{Bourb1, Bourb2, FuH}
\begin{proposition}
\label{PR_N_j_irrep}
The subspaces $P^k_{k-2j}\ss \Lambda^k(\HH_N)$ are invariant and irreducible under the $\Sp(\HH_N, \omega_N) \ss \Gl(\HH_N)$
action on $\Lambda^k(\HH_N)$, which decomposes into their direct sum:
 \be
 \Lambda^k(\HH_N) = \bigoplus_{j=0}^{\lfloor k/2 \rfloor} P^k_{k-2j}.
 \ee
 The exterior algebra $\Lambda(\HH_N)$ thus decomposes into the direct sum:
 \be
 \Lambda(\HH_N) = \bigoplus_{k=0}^N\bigoplus_{j=0}^{\lfloor k/2 \rfloor} P^k_{k-2j}.
    \label{eq:LambdaWDecompositionIrreducibles}
 \ee
The isomorphism class of  $P^k_{k-2j}$ is given by the partition $(1)^{k-2j}$, and has dimension
 \be
 \binom{2k}{k-2j} -  \binom{2k}{k-2j-2} =\frac{2j+1}{k+1}\binom{2k +2}{ k-2j}.
 \label{dimP_k_k-2j}
 \ee
 \end{proposition}

 The fact that a subspace $w^0\ss \HH_N$ is Lagrangian is defined by the  condition
 \be
\omega_N (u,  v) = 0, \quad \forall \ u, v \in w^0,
 \ee
 which implies:
 \begin{proposition}
 \label{plucker_lagrangian_linear_image}
 The image of the restriction of the Pl\"ucker map
 \be
 \grP\grl_V: Gr^\LL_V(\HH_N,\omega_N))\ra\Pb(\Lambda^N(\HH_N))
 \ee
to the Lagrangian Grasssmannian $Gr^\LL_V(\HH_N,\omega_N))$ lies in $P^N_N$,
which is its linear span.  The number of independent linear relations (\ref{linear_lagrange_rels}) that determine it is thus $\binom{2N}{N-2}$.
  \end{proposition}
  \begin{remark}
  A simple way to express these linear relations in terms of Pl\"ucker coordinates is given in \cite{CZ}.
   For any subset $\alpha \ss \{-N, \dots,  N-1\}$ of cardinality $N-2$,
   whose negative indices are denoted $-I$ and nonnegative indices $J-1$, let
  $ \lambda(I(\alpha, i), J(\alpha, i) )$  denote the partition obtained by adding the pair $(i, i) $ to
the increasingly ordered sets $(I,J)$,  where $i$ does not belong to $I \cup J$. The linear
 relations determining $P^N_N$ are then expressed in terms of the Pl\"ucker coordinates by
  \be
  \sum_{}\pi_{(\lambda(I(\alpha, i)), (J(\alpha, i))} =0,
  \label{linear_rels}
  \ee
  where the sum is over all $i\in \{1, \dots, N\}$ that do not belong to $I\cup J$.
    \end{remark}
 In particular all basis elements
 \be
|\lambda\rangle = |(\ab | \ab) \rangle := e_{(I, I)}
\ee
 corresponding to symmetric partitions  $\lambda = \lambda^T$ belong to $P^N_N$.
 Their linear span may be viewed as a subspace of  $P^N_N$ or, equivalently, as the exterior
 space $\Lambda(V)$, under the injection $i_{\Lambda(V)}$ defined in (\ref{LambdaV_inject}}).

\begin{corollary}
Of these linear relations, it is possible to choose
\be
\sum_{j=1}^{\lfloor N/2 \rfloor}(-1)^{j - 1}\binom{2 N} {N - 2 j} = \frac{1}{2}\binom{2N}{N}  - 2^{N-1}
\label{number_2_term}
\ee
independent ones involving two terms only,  consisting of the equalities
\be
\pi_\lambda(w^0) = \pi_{\lambda^T}(w^0)
\label{2_term_linear_rels}
\ee
between Pl\"ucker coordinates corresponding to all  pairs $(\lambda, \lambda^T)$
of distinct partitions within the square $(N)^N$.
\end{corollary}

\subsection{Lagrange map, hyperdeterminantal relations and inverse}
\label{lagrange_hyperdet_relations_N}

\subsubsection{Lagrange coefficients, principal minors,  hyperdeterminantal relations}
\label{lagrange_princip_minors}

On the big cell, the  {\em hyperdeterminantal relations} \cite{HoSt, Oed}  are satisfied  by the  principal minor determinants
of the affine coordinate matrix which, up to projectivization, coincide with the Lagrange coefficients
   \be
   D_J(A(w^0) := \det(A_J(w^0))= \LL_J(w^0).
   \ee
   To express these, we extend the definition of $\LL_J$ to allow any distinct subset $J=(J_1, \dots , J_r)\ss (1, \dots, N)$
   of cardinality $r$, regardless of order, with the value of $\LL_J$ the same for all orderings. Now choose an additional
   triplet $(j_1, j_2, j_3)$ of distinct elements of $\{1, \dots,  N\}$ which are also distinct from the elements of $J$ (so $r\le N-3$).
   and denote by $(J, j_a)$, $(J, j_a, j_b)$ and $(J, j_1, j_2, j_3)$, for $a, b = 1, \dots 3,\ a\neq b$, the subsets of $(1, \dots , N)$
   with the indicated elements. We then have the following result:
   \begin{proposition}
   \label{hyperdet_finite}
   The Lagrange coefficients  satisfy the ``core'' hyperdeterminantal relations
   \bea
 &\& \LL_J^2\LL^2_{J, j_1,j_2,j_3} +\LL^2_{J, j_1} \LL^2_{J, j_2, j_3}  +\LL^2_{J, j_2} \LL^2_{J, j_1, j_3} +\LL^2_{J, j_3} \LL^2_{J, j_1, j_2}  \cr
&\&-2\LL_J \LL_{J,j_1}\LL_{J, j_2, j_3}\LL_{J, j_1. j_2, j_3}  - 2\LL_J  \LL_{J,j_2}\LL_{J, j_1, j_3}\LL_{J, j_1. j_2, j_3}  -2\LL_J  \LL_{J,j_3}\LL_{J, j_1, j_2}\LL_{J, j_1. j_2, j_3} \cr
 &\& -2\LL_{J,j_1}\LL_{J, j_2}\LL_{J, j_1,j_3}\LL_{J,j_2, j_3}  -2\LL_{J, j_1}\LL_{J, j_3}\LL_{J, j_1,j_2}\LL_{J,j_2, j_3}
 - 2\LL_{J,j_2}\LL_{J, j_3}\LL_{J, j_1,j_2}\LL_{J, j_1, j_3} \cr
 &\&    -4 \LL_{J}\LL_{J,j_1, j_2}\LL_{J, j_1, j_3} \LL_{J, j_2, j_3} - 4 \LL_{J, j_1}\LL_{J, j_2} \LL_{J, j_3} \LL_{J, j_1, j_2, j_3} =0.
 \cr
 &\&
 \label{hyper_det_rels}
 \eea
   \end{proposition}

     \begin{remark}
     This result is proved  in \cite{HoSt, Oed} for the principal minor determinants of any symmetric
   $N\times N$ matrix. We give another proof in Sections \ref{geometry_plucker_restriction},  \ref{hyperdet_Gr_L_3_6}, based on combining
   the Pl\"ucker relations for any  element  $w^0\in Gr_V(\HH_N)$ with the linear relations that assure it belongs
   to the Lagrangian Grassmannian $w^0\in \Gr^\LL_V(\HH_N, \omega_N)$. In  \cite{Oed}, it was shown that the image of
   the principal minors map is cut out by  the orbit of the ``core'' quartic hyperdeterminantal relations (\ref{hyper_det_rels}),
   under the subgroup
   \be
   G_N:=   (\Sl(2))^{ N}  \rtimes  S_N   \subset  \Sp(\HH_N, \omega_N),
   \ee
   where the $\Sl(2)$'s act within the planes $\{e_{-i}, e_{i-1}\}_{i=1, \dots, N}$ and  $S_N$ by permuting them.

  In Section \ref{hyperdet_Gr_L_3_6}, we identify the  eight distinct principal
     minors of size $(r+a) \times (r+a)$, for $a=0,1, 2, 3$  that correspond to the nonzero
     columns and rows appearing in (\ref{hyper_det_rels}). These are all of the same form as the single
     quartic relation satisfied by the eight symmetric Pl\"ucker  coordinates for $ \Gr^\LL_{\Cb^3}(\Cb^3 \oplus \Cb^{3*}, \omega_3)$.
     By varying the choice of ($J, j_1, j_2, j_3)$ as subsets of $\{1, \dots ,N\}$, we obtain
     the core hyperdeterminantal relations (\ref{hyper_det_rels}).
        \end{remark}

   \subsubsection{Inverse of the Lagrange map}
\label{invert_lagrange_map}
The Lagrange map (\ref{lagrange_map_finite_def}) is constant on the orbits of the subgroup
\be
(\Zb_2)^N  =  \II_\epsilon := \{\diag( \epsilon_{-N}, \dots  , \epsilon_{-1}, \epsilon_{0}, \dots, \epsilon_{N-1})\}, \ss\\Sp(\HH_N, \omega_N),
\ee
where
\be
\epsilon_{-i} = \epsilon_{i-1} = \pm 1, \quad i=1, \dots, N,
\ee
consisting of any number of reflections inside the canonical coordinate $2$-planes  \break \hbox{$\{e_{-i}, e_{i-1}\}_{i= 1, \dots, N}$},
since it leaves invariant the Pl\"ucker coordinates $\pi_\lambda(w^0)$ for all symmetric partitions $\lambda= \lambda^T$.
In fact, the converse is also true \cite{VGM}; two elements of $\Gr^\LL_V(\HH_N, \omega_N)$ have the same image under the
Lagrange map if and only if they lie on the same  $(\Zb_2)^N$ orbit.
Generically, $(\Zb_2)^{N}$ has the $2$-element subgroup $\{\pm\Ib_{2N}\}$ as stability subgroup,  and
there is an open dense stratum in which all the orbits have $2^{N-1}$ elements.
 But there are strata consisting of orbits of all cardinalities $2^k$, for $0 \le k \le N-1$,
so the quotient by this group action is not a manifold, but an orbifold.

As with the Pl\"ucker coordinates \cite{GH},  the Lagrange coefficients $\{\LL_J\}$
may be interpreted as  holomorphic sections of a line bundle: the  (dual) determinantal line bundle $\Det^{*} \ra \Gr^\LL_V(\HH_N, \omega_N)$,
defined as the pullback, under the Lagrange map, of the hyperplane section bundle $\OO(1) \ra \Pb(\Lambda(V))$
\be
\Det^* \ra \Gr^\LL_V(\HH_N, \omega_N) := \LL_N^*( \OO(1) \ra \Pb(\Lambda(V))).
\ee
Although this is equivalent to the restriction of the dual determinantal line bundle $\Det^* \ra \Gr_V(\HH_N)$ to
$\Gr^\LL_V(\HH_N,  \omega_N)$,  the sections corresponding to symmetric partitions span
$\Lambda^*(V)$, which  is realized as the $2^N$ dimensional subspace of $\Lambda^{N*}(\HH_N)$
defined by the injection map (\ref{LambdaV_inject}) or, equivalently, by the basis elements corresponding to
symmetric partitions.

\subsection{The geometry of Plucker relations. Restriction to Lagrangian Grassmannians}
\label{geometry_plucker_restriction}

Recall that the  image of the Lagrangian Grassmannian $\Gr^\LL_V(\HH_N, \omega_N)$ under the Pl\"ucker map
 is cut out in $\Pb(\Lambda^N(\HH_N))$  by the combination of the Pl\"ucker relations, corresponding to a decomposable $N$--vector
 defining the $N$-plane and the linear relations following from the fact that the $N$-plane is a Lagrangian subspace.

We can restate the Pl\"ucker relations as follows. They are determined by first  choosing a ``seed'' multi-index $I_0$
   of cardinality  $i_0<k-1$, and  then completing it with indices $i_1, ..,i_{k-1-i_0}$ to a multi-index $I$ of cardinality $k-1$, and with $j_1,..,j_{k+1-i_0}$
   to  a multi-index $J$ of cardinality $k+1$, in such a way that all the added indices are distinct.
   The corresponding Plucker relation on $\Pl^n_k(w)$ is then
\be
 \sum_{s=1}^{k+1-i_0} (-1)^{s} \tilde \pi_{I_0, i_1,i_2,..,i_{k-1-i_0},j_s} \tilde \pi_{I_0, j_1,..,\hat j_s,..,j_{k+1-i_0}} =0.
 \ee
Thus, the number of terms in the sum is $k+1-i_0$.
The  ``short'' Pl\"ucker relations, occur when $i_0= k-2$, and consist of a three term sum
\be
 \label{shortP1}
\sum_{cycl(j,k,\ell)} \tilde \pi_{I^0,i,j}\tilde \pi_{I^0,k,\ell} = 0.
\ee

We can show, by restriction and projection to subspaces, that a combination of suitably chosen 
 short Pl\"ucker relations with the isotropy condition for Lagrangian Grassmannians imply the full set of Pl\"ucker relations with the isotropy condition, at least on a generic locus. We will see in the next section that the short relations and isotropy then determine the hyperdeterminantal conditions.   The idea, roughly, is to intersect with a family of coordinate subspaces. To illustrate this, first consider the corresponding statement, for the ordinary Grassmannian of $k$-planes in $n$-space, that the short (three-term) Plucker relations 
 determine the Grassmannian, on a generic locus.  This is  proved in \cite{HB}, App. D  (cf. also \cite{KPRS}), using determinantal identities.   To see it geometrically, note that for the multi-index $I_0$, intersection with a subspace determined by the complementary  multiindex $I_0^c$ is simply given on the level of the exterior algebra by contraction by $f_{I_0}$, as long as this contraction gives a non-zero result (this is the necessary genericity). The Pl\"ucker relation "survives"  this operation simply by removal of the $I_0$, and becomes a relation for smaller dimensional planes in a smaller dimensional space. In particular, if  $i_0= k-2$, so that the Pl\"ucker relation has length three, the relation becomes one for a two dimensional space  in $n-k+2$ space, if $n$ is our initial dimension. But for these the only Pl\"ucker relations are short. Doing this for all possible choices of $I_0$ of length $k-2$, we find:
 
 \begin{proposition}
Let $\phi$ be a $k$-vector  in $\Lambda^k(U)$  which, is {\em generic}, in the sense of  belonging to the Zariski open set  on which the
contractions $i_{f_{I^0}}(\phi)$ are non-zero. If when contracted with  every coordinate $k-2$-vector it defines a $2$-plane,
(i.e. is a decomposable $2$-vector), then $\phi$ defines a $k$-plane (i.e., is decomposable).
In consequence, the three-term Pl\"ucker relations for the $k$-vector imply the full set of Pl\"ucker relations.
 \end{proposition}
 The proof consists in taking local coordinates, and then  in looking at all these two-planes, checking that they are in a suitable set compatible, and then piecing together the result into a $k$-plane.

We now consider a similar question for Lagrangian subspaces. Here,  there will be a family of six dimensional spaces obtained from the coordinate subspaces. 
However, the procedure will not just be one of intersection, but rather intersection followed by projection. (Note that this is completely
 in the spirit of symplectic reduction, where one first restricts to a subvariety, then quotients by a null foliation.) Define the basis
 \be
\{ f_1, f_2, ..., f_N, f_1^*, f_2^*, ..,f_N^*\}
\label{f_basis}
 \ee
  of the $2N$-dimensional symplectic space $\HH_N$ by
  \be
  f_i := e_{-i}, \quad f^*_{i} := (-1)^{i-1} e_{i-1}, \quad 1\le i \le N,
  \label{e_f_basis_change}
  \ee
In this basis, the symplectic form is
   \be
   \omega_N = \sum_j f_j\swedge f^* _j.
   \ee
   Denote the components of a vector $v\in \HH_N$ relative to this basis as $\{a_i, a_i^*\}_{i=1, \dots, N}$,
   \be
   v = \sum_i (a_if_i +a_i^* f_i^*).
   \ee
   
   The corresponding basis  $\{f_K\}$ for  $\Lambda^N(\HH_N)$  is given by
   \be
   f_K = f_{K_1} \swedge \cdots  f_{K_N},
   \label{f_K_basis}
   \ee
   where $K$ is a multi-index $(K_1,K_2,..K_N)$ with distinct, increasingly ordered $K_i$'s, first of type $j$, 
   followed by those of type $j^*$, $j\in \{1, \dots, N\}$. Relative to this basis, any $N$-vector $\phi\in \Lambda^N(\HH_N)$
   may be expressed as
   \be
\phi = \sum_{K} \tilde\pi_K f_K.
   \ee
Note that for $K$  to correspond to a symmetric partition, in the notation of the preceding sections, 
   means that, for all $j\in \{1, \dots, N\}$, $K=(K_1, \dots, K_N)$ contains either $j$ or $j^*$, but not both. 
This is equivalent to the corresponding decomposable $N$-vector $f_K$  being the Pl\"ucker image of  a Lagrangian (i.e., maximal isotropic) subspace.
 We also  will require  basis multi-vectors for $\Lambda^k(\HH_N)$ of degree $k<N$, which satisfy the symmetry condition
  that they contain  either $j$ or $j^*$, but not both. The corresponding multi-indices, viewed as subsets of  $\{i, i^*\}_{i=1, \dots, N}$ are defined as follows.
Let $I= (I_1, .., I_k) \ss (1,..,N)$ be an increasingly ordered subset, viewed as a multi-index of cardinality $k$.
We then choose a function  $A$ on the space of such $k$-indices which, to each $I_j$ associates either $A(I_j) = I_j$ or $I_j^*$.
Let $(A, I)$ denote the corresponding ``marked'' multi-index of cardinality $k$
     \be
     \label{marked_index}
     (A,I) = (A(I_1), A(I_2), ...A(I_k))
     \ee
      (written in their correct order) and denote the corresponding basis multivectors $f_{(A,I)}$. For each $(A, I)$, there is a complementary $(B,I)$ consisting of
  the complementary elements; i.e.  $B(I_j) = I^*_j \ (\mathrm{resp.\ }I_j)$ if and only if $A(I_j) = I_j\  (\mathrm{resp.\ }I^*_j)$.   Let  
      \be
      p_{(B,I)}: \HH_N \ra  \HH_N
      \label{p_B_I_proj}
      \ee
      denote the projection map onto the $(N+k)$-dimensional subspace  spanned by the
      basis vectors complementary to $\{f_{B(I_1)} \dots, f_{B(I_k)}\}$, with kernel the space  $\span \{f_{B(I_1)} \dots, f_{B(I_k)}\}$,
   and let
      \be
      p^j_{(B,I)}: \Lambda^j(\HH_N) \ra  \Lambda^j(\HH_N), \quad j\in \{1, \dots, N\}
      \label{p_B_I_proj_j}
      \ee
     denote the lift of this map to $\Lambda^j(\HH_N)$.
     
     Our operations  will be:  contraction $i_{f_{(A,I)}}\phi$ of $\phi\in \Lambda^N(\HH_N)$  with $f_{(A,I)}$ 
 (so intersection with the co-isotropic plane corresponding to $\tilde \pi_{(A,I)}= 0$), giving an element of $\Lambda^{N-k}(\HH_N)$, 
 followed by projection  $p^{N-k}_{(B,I)}$  to $\Lambda^{N-k}(p_{(B,I)}(\HH_N))\subset \Lambda^{N-k}( \HH_N)$ (i.e., setting the corresponding coordinates to zero.) 
  Note that if we define  the contraction and projection in such a way that we stay  in $\HH_N$, the two operations commute, in the sense that
  \be
  p^{N-k}_{(B,I)}\circ i_ {f_{(A,I)}} = i_ {f_{(A,I)}}\circ  p^N_{(B,I)}
\ee
 Now let ${(A,B,I)^c}$ denote the set of $2N-2k$ indices in the complement of the union of the multi-indices $(A,I), (B, I)$;
 i.e., the indices $\{j, j^*\}_{j\in I^c}$, and let $\HH_{(A,B,I)^c}$ be the $2N-2k$ dimensional space spanned by the vectors 
 with indices in ${(A,B,I)^c}$. Thus
 \be
 p^{N-k}_{(B,I)}\circ i_ {f_{(A,I)}}(\phi) = i_ {f_{(A,I)}}\circ  p^N_{(B,I)}(\phi) \in \Lambda^{n-k}
\in  \Lambda^{N-k}(\HH_{(A,B,I)^c}) \subset \Lambda^{N-k}( \HH_N)
 \ee
  gives us an $N-k$-vector in  $\Lambda^{N-k}(\HH_{(A,B,I)^c})$ for any $\phi\in \Lambda^N(\HH_N)$.

     We now fix $k=N-3$.   
The  result of the contraction and projection is  now a $3$-vector on the  $6$-dimensional subspace   $\HH_{(A,B,I)^c} \ss \HH_N$. 
Applying this to  a decomposable isotropic element $ \phi \in \Lambda^N (\HH_N))$, the resulting $3$-vector  is again decomposable 
and isotropic in $\Lambda^3(\HH_{(A,B,I)^c})$. We also have the converse:
\begin{proposition}
\label{prop:plucker-hyper}
A generic element $\phi \in \Lambda^N (\HH_N)$ is the Pl\"ucker image of a
Lagrangian plane $w^0\in \Gr^\LL_V(\HH_N)$ if and only if, for all $(A, I)$ with $I$ of cardinality $N-3$,
 the elements 
 \be
 p^3_{(B,I)}( i_{f_{(A,I)}}(\phi)) = i_{f_{(A,I)}}(p^N_{(B,I)}(\phi))
 \ee
  represent null (isotropic)  $3$-planes in $\Lambda^3(\HH_{(A,B,I)^c})$.
\end{proposition}
The proof proceeds in essence in taking all of these Lagrangian three-planes, and seeing that the fact that they come from a common element $\phi$ allow us to piece them together into a Lagrangian $n$-plane. Again, the genericity required is that the contractions and projections give non-zero results.

Now consider what this means in terms of the coordinates $\tilde \pi_K$ of the original $N$-vector 
 $\phi \in \Lambda^N(\HH_N)$, and Pl\"ucker relations for the $  p^3_{(B,I)}i_{f_{(A,I)}}(\phi)$.  
 The  relations on the $3$-planes  are given by taking  the  multi-index $(A,I)$ of cardinality $N-3$ as ``seed''.
 This  is completed in turn by adding to $(A,I)$ first  two indices $(L_1, L_2)$, giving  a multi-index $L$ of cardinality
  $N-1$, and then  four indices $(K_1, K_2, K_3, K_4)$, giving a multi-index $K$ of  cardinality $N+1$. 
  We require that these  extra indices $L, K$ avoid the elements of $(B, I)$. They thus lie in the set
 $\{i, i^*, j, j^*, k, k^*\}$, where $i,j,k$ are the three indices not in $I$, and $(L_1, L_2)$ and $(K_1, K_2, K_3, K_4)$
  can overlap by at most one element. (If they overlap by one element, we add that to the seed.)
  The corresponding Pl\"ucker relations are then
\be
\label{p1}
 \sum_{s=1}^{4} (-1)^{s} \tilde \pi_{(A,I),L_1,L_2, K_s} \tilde \pi_{(A,I), K_1,..,\hat K_s,..,K_4} =0
 \ee
for no overlap and
\be
\label{p2}
\sum_{s=2}^{4} (-1)^{s} \tilde \pi_{(A,I),L_1,L_2,K_s} \tilde \pi_{(A,I),L_1,..,\hat K_s,..,K_4} =0
\ee
when $L_1 = K_1$.
Varying $L,K$ gives the equations for the $3$-plane in  the $6$-plane $\HH_{(A,B,I)^c}$ corresponding to $(A,I)$, 
essentially by a correspondence $ \tilde \pi_{(A,I),\mu,\nu,\sigma}\leftrightarrow \tilde \pi_{ \mu,\nu,\sigma}$. 
The  Pl\"ucker relations restrict to Pl\"ucker relations on $\Gr_3(\HH_{(A,B,I)^c})$, for an appropriate choice of indices.  
Restriction of the isotropy condition is simpler; we just require  that contractions with the restriction 
\be
\omega_{I}:= \omega_N\vert_{\HH_{(A,B,I)^c}}
\ee
 of the symplectic form to $\HH_{(A,B,I)^c}$ give zero. Thus:

\begin{proposition} 
\label{prop:N_3_hyper_red}
For generic $\phi$, the Pl\"ucker relations, together with the symplectic isotropy conditions on $\phi$, 
are equivalent to the relations (\ref{p1}), (\ref{p2}) for  all $(A,I)$) with $I$ of cardinality $N-3$,
  together with the isotropy conditions 
  \be
  i_{\omega_{I}}((p^3_{(B,I)}i_{f_{(A,I)}}(\phi))) =0,
    \ee
    where $p^3_{(B,I)}(i_{f_{(A,I)}}(\phi))$ is viewed as an element of $\Lambda^3(\HH_{(A,B,I)^c)})$.
\end{proposition}


\subsection{Hyperdeterminantal relations for $\Gr^\LL_{\Cb^3}(\Cb^3\oplus \Cb^{3*}, \omega_3)$ }
\label{hyperdet_Gr_L_3_6}

We have thus reduced the problem, at least on an open dense set, to a family of Pl\"ucker relations and  isotropy conditions in dimensions $(3,6)$; that is, for
elements of $\Lambda^3(\Cb^3\oplus \Cb^{3*})$ corresponding to isotropic 3-planes. Our aim is to now combine these into one relation,
the hyperdeterminantal relation, for each of these $3$-planes.

For $i,j,k\in \{1,2,3,1^*,2^*,3^*\} $, let $\tilde \pi_{ijk}$,  in the indicated order, denote the Pl\"ucker coordinates of a $3$-dimensional
subspace $w^0 \ss \Cb^3\oplus \Cb^{3*}$,  viewed as an element of the Grassmannian $\Gr_{\Cb^3}(\Cb^3\oplus \Cb^{3*}, \omega_3)$
whose Pl\"ucker image, up to projectivization, is given by
\be
\phi :=  \Pl^{2N}_N(w^0) = \sum_{i,j,k\in \{1,2,3,1^*,2^*,3^*\} } \tilde \pi_{ijk}f_i\wedge f_j\wedge f_k   \in \Lambda^3(\Cb^3\oplus \Cb^{3*}).
\ee
This gives $20$ projective coordinates, and so $19$ parameters.  Eight of these correspond to symmetric partitions:
\bea
S_0 &\&:=\tilde{\pi}_{123}, \quad S_{1} := \tilde \pi_{2 3 1^*}, \quad  S_{2} := -\tilde \pi_{13 2^*}, \quad S_{3} := \tilde \pi_{1 2 3^*}, \cr
S_{0^*} &\&:= \tilde \pi_{1^*2^*3^*}, \quad S_{1^*} := \tilde \pi_{1 2^*3^*}, \quad  S_{2^*} := - \tilde \pi_{2 1^*3^*}, \quad S_{3^*} := \tilde \pi_{3 1^*2^*},
\label{symmetric_8}
\eea
in which the $\tilde \pi_{ijk}$ are chosen such that   $i = 1$ or $1^*$, $j = 2$ or $2^*$, $k=3$ or $3^*$.
The remaining $12$  ``nonsymmetric'' Pl\"ucker coordinates form, by the linear Lagrange conditions,
 six equal pairs, which  are labelled  by mutually dual partitions
\bea
T_1 &\&:= \tilde \pi_{12 2^*} = - \tilde \pi_{1 3 3^*} ,  \quad T_2:=\tilde \pi_{233^*} = \tilde \pi_{121^*} ,  \quad T_3 :=  \tilde \pi_{232^*}=- \tilde \pi_{131^*} , \cr
T_{1^*}&\&:=\tilde \pi_{ 21^*2^*} = - \tilde \pi_{3 1^* 3^*},    \quad T_{2^*} := \tilde \pi_{3 2^* 3^*} = \tilde{\pi}_{1 1^* 2^*} ,
\quad T_{3^*} := -\tilde \pi_{11^*3^*} = \tilde \pi_{22^*3^*}.
\label{unsymmetric_12}
\eea
There are $120$ three term (``short'') Pl\"ucker relations:
\be
\sum_{\nu = 1}^3 (-1)^\nu \tilde \pi_{i_1, i_2 ,j_\nu} \tilde \pi_{i_1, j_1, ..., \widehat{j_\nu},...,j_3} = 0,
\ee
with five  distinct indices $(i_1, i_2, j_1, j_2, j_3)$, and  $15$ four term ones
\be
 \sum_{\nu = 0}^3 (-1)^\nu \tilde \pi_{i_1,i_2,j_\nu} \tilde \pi_{j_0, ..., \widehat{j_\nu},...,j_3} =0,
\ee
 with  six distinct indices $(i_1, i_2, j_0, j_1, j_2, j_3)$.
 (These are obviously very redundant, since the Lagrangian Grassmannian has dimension $6$.)

 We can eliminate the non-symmetric coordinates from some of the Pl\"ucker relations
  to obtain one quartic relation for the remaining $8$ symmetric coordinates which, in addition to projectivization, cuts us down to 6 dimensions,
and so gives the isotropic Grassmannian, at least on an open set.
The two short Pl\"ucker relations
\begin{subequations}
\bea
 \tilde \pi_{123}\tilde \pi_{12^*3^*} -  \tilde \pi_{122^*}\tilde \pi_{13 3^*} +  \tilde \pi_{123^*}\tilde \pi_{132^*}  &\& =0, \\
 \tilde \pi_{23 1^*}\tilde \pi_{1^*2^*3^*} -  \tilde \pi_{21^*2^*}\tilde \pi_{3 1^*3^*} +  \tilde \pi_{2 1^*3^*}\tilde \pi_{3 1^*2^*} &\& =0
\eea
\end{subequations}
give
\begin{subequations}
\bea
\label{short1}
 T_{1}^2&\&=  - S_{0 }S_{1^*}   +  S_{2} S_{3} , \\
\label{short1*}
 T_{1^*}^2 &\& = - S_{0^*} S_{1} -   S_{2^*}S_{3^*}
\eea
Similarly, we have
\bea
\label{short2}
 T_{2}^2&\&=  -S_{0} S_{2^*}   +  S_{1} S_{3} , \\
\label{short2*}
 T_{2^*}^2 &\& =  -S_{0^*} S_{2}- S_{1^*}S_{3^*} \\
 \label{short3}
 T_{3}^2&\&=  -S_{0} S_{3^*}   + S_{1} S_{2} , \\
\label{short3*}
 T_{3^*}^2 &\& =  -S_{0^*} S_{3} - S_{1^*}S_{2^*}
\eea
\end{subequations}
The four term relations
\begin{subequations}
\bea
 \tilde \pi_{123}\tilde \pi_{1^*2^*3^*}-  \tilde \pi_{121^*}\tilde \pi_{3 2^*3^*} + \tilde \pi_{122^*}\tilde \pi_{3 1^*3^*} - \tilde \pi_{123^*}\tilde \pi_{ 31^*2^*} &\&=0,\\
- \tilde \pi_{231^*}\tilde \pi_{12^*3^*}-   \tilde \pi_{12 1^*}\tilde \pi_{3 2^*3^*} - \tilde \pi_{21^*2^*}\tilde \pi_{ 13 3^* } + \tilde \pi_{21^*3^*}\tilde \pi_{ 132^*} &\&=0
\eea
\end{subequations}
give
\begin{subequations}
\bea
T_1 T_{1^*} + T_2 T_{2^*}&\&= S_0 S_{0^*} - S_3 S_{3^*}, \\
T_1 T_{1^*} - T_2 T_{2^*} &\&= S_1 S_{1^*} -  S_2 S_{2^*},
\eea
\end{subequations}
and hence
\begin{subequations}
\bea
2T_1 T_{1^*}&\&=  S_0 S_{0^*} + S_1 S_{1^*} - S_2 S_{2^*}  - S_3 S_{3^*} ,	
\label{long1} \\
2T_2 T_{2^*}  &\&= S_0 S_{0^*}   - S_1 S_{1^*}  + S_2 S_{2^*}  - S_3 S_{3^*} ,
\label{long2}
\eea
and similarly, we have
\be
2T_3 T_{3^*}=  S_0 S_{0^*} - S_1 S_{1^*} - S_2 S_{2^*}  + S_3 S_{3^*}.	
\label{long3}
\ee
\end{subequations}

Squaring (\ref{long1}) and equating this to the product of the expressions   (\ref{short1}), (\ref{short1*}) gives
\bea
&\& S_0^2S_{0^*}^2  + S_{1}^2S_{1^*}^2  +S_2^2S_{2^*}^2  +S_3^2S_{3^*}^2
=2S_0S_{0^*} S_{1}S_{1^*} + 2S_0S_{0^*}S_2S_{2^*}+2 S_0S_{0^*} S_{3}S_{3^*}
 \cr
 &\&  + 2 S_{1}S_{1^*} S_{2}S_{2^*}  + 2 S_{1}S_{1^*}  S_3S_{3^*}+ 2  S_{2}S_{2^*} S_3S_{3^*}
 - 4S_{0^*}S_1  S_2S_3 - 4S_{0}S_{1^*} S_{2^*}S_{3^*} \cr
 &\&
  \label{hyperdet_3_6}
\eea
or, equivalently,
\bea
&\&  \tilde{\pi}_{123}^2\tilde{\pi}_{1^*2^*3^*}^2  +\tilde{\pi}_{123^*}^2\tilde{\pi}_{ 31^*2^*}^2 + \tilde{\pi}_{231^*}^2\tilde{\pi}_{12^*3^*}^2  +\tilde{\pi}_{21^*3^*}^2\tilde{\pi}_{  132^*}^2 \cr
&\& \quad =2  \tilde{\pi}_{123}\tilde{\pi}_{1^*2^*3^*}  \tilde{\pi}_{123^*}\tilde{\pi}_{ 31^*2^*} +2 \tilde{\pi}_{123}\tilde{\pi}_{1^*2^*3^*}  \tilde{\pi}_{231^*}\tilde{\pi}_{12^*3^*}
+2 \tilde{\pi}_{123}\tilde{\pi}_{1^*2^*3^*} \tilde{\pi}_{21^*3^*}\tilde{\pi}_{132^*} \cr
 &\&  \qquad+ 2 \tilde{\pi}_{123^*}\tilde{\pi}_{ 1^*2^*3}  \tilde{\pi}_{23 1^*}\tilde{\pi}_{12^*3^*}  + 2 \tilde{\pi}_{123^*}\tilde{\pi}_{ 3 1^*2^*}  \tilde{\pi}_{2 1^*3^*}\tilde{\pi}_{ 132^*}+ 2  \tilde{\pi}_{231^*}\tilde{\pi}_{12^*3^*} \tilde{\pi}_{21^*3^*}\tilde{\pi}_{  13 2^*}\cr
  &\&\qquad  + 4\tilde{\pi}_{123}\tilde{\pi}_{12^*3^*}  \tilde{\pi}_{21^*3^*}\tilde{\pi}_{31^* 2^*} + 4\tilde{\pi}_{123^*}\tilde{\pi}_{13 2^*} \tilde{\pi}_{231^*}\tilde{\pi}_{1^*2^*3^*},
  \label{hyperdet_3_6}
\eea
which is the single hyperdeterminantal relation for $\Gr^\LL_{\Cb^3}(\Cb^3 \oplus \Cb^{3*}, \omega_3)$.
(The same relation may be derived  {\em mutatus mutandis} using the pairs $(T_2, T_{2^*})$ or  $(T_3, T_{3^*})$.)

As explained in Section \ref{lagrange_map}, on the big cell of the Lagrangian Grassmannian, the
$N$-dimensional subspace $w^0$ is represented as the graph of a map $A(w^0):\Cb^N\rightarrow (\Cb^N)^*$, given by the affine
coordinate  matrix $A(w^0)$, which is symmetric,  and the Pl\"ucker coordinates  corresponding to symmetric partitions are
projectively equivalent  to its principal minors.
 Relation (\ref{hyperdet_3_6}) is an example of the ``core'' hyperdeterminantal relations  studied in \cite{HoSt, Oed}.
 In  \cite{Oed}, it was shown that these relations, orbited by the group $G_N$ defined in (\ref{G_N_def})
 as the semi-direct product of $\Sl(2,\Cb)^N$  with the symmetric
 group on $N$ letters, where the $\Sl(2,\Cb)$'s act within the $2$-planes spanned the dual pairs $(f_i, f_i^*$), and $S_N$
 permutes them, cut out the variety defined by the principal minors of $A$.

  Note that the symmetric  partition Pl\"ucker coordinates do not quite determine the isotropic plane.
  As explained for general $N$ in Subsection \ref{invert_lagrange_map}, the short Pl\"ucker relations (\ref{short1}) - (\ref{short3*})
  only determine the non-symmetric    coordinates $(T_1, T_2, T_3, T_{1^*}, T_{2^*}, T_{3^*})$ up to the action of the group $(\Zb_2)^3$ of sign changes within
  the canonical coordinate planes, which replaces these by  $(\epsilon_1T_1, \epsilon_2T_2, \epsilon_3T_3, \epsilon_1T_{1^*}, \epsilon_2 T_{2^*},
  \epsilon_3T_{3^*})$ for $\{\epsilon_i= \pm 1\}_{i=1, 2, 3}$. The hyperdeterminantal relation cuts out the variety obtained
  as the image of any of the points on an orbit, which, generically, is of cardinality $8$.

On $\HH_N = V\oplus V^*$, the definitions of the multi-indices $(A,I), (B,I)$  may be 
adapted to the basis $(e_{-N}, \dots, e_{N-1})$ as follows. The multindex $(A,I)$ is defined by combining 
$I= (I_1, ...,I_{N-3})\subset (1,2,..,N)$ as before with the ``marking'' function $A$ that associates to each $I_j$ either
 $ A(I_j) = -I_j$ or $I_j-1$ giving
\be
(A,I) = (A(I_1), A(I_2), ...,A(I_{N-3})
\label{marked_index2} 
\ee
 In the same way, we define the complementary  assignments $B(I_j) = -I_j$
  (resp. $I_j-1$) if $A(I_j) = I_j-1$ (resp. $-I_j$),  and  the complementary marked multi-index $(B,I)$.
The operators $p^3_{(B,I)}i_{f_{(A,I)}}$ are given, {\em mutatis mutandis},
 by the same operations of contraction and projection.   With thus have:
 \begin{proposition}
A generic element $\phi \in \Lambda^N(\HH_N)$ represents a  Lagrangian plane if and only if
 the eight symmetric coordinates of all the elements $ p^3_{(B,I)}i_{f_{(A,I)}}(\phi)$ satisfy the ``core''
 hyperdeterminantal relations (\ref{hyper_det_rels}).
 
\end{proposition}

Finally, if we let the multi-index $J$ denote the set of indices $i\in I $ for which $A(i) = i-1$,  the symmetric Pl\"ucker coordinates 
of $p^3_{B,I}(i_{f_{(A,I})}(\phi))$ are precisely the Lagrangian coefficients 
$\mathcal L_{J}, \mathcal L_{J, j_1},  \mathcal L_{J, j_1, j_2}, $ $\mathcal L_{J, j_1, j_2,j_3}$
of the image of the Lagrange map. Therefore  a generic  element of $\HH_N $ represents a Lagrangian plane if and only if the image
of the Lagrange map satisfies the hyperdeterminantal relations  ({\ref{hyper_det_rels}) of Proposition \ref{hyperdet_finite}.


\subsection{Hexahedron recurrence equations}
\label{sec:hexahedron}

   The hyperdetermantal relations (\ref{hyper_det_rels}) were introduced as
 integrable systems of recurrence relations on lattices by Kashaev \cite{Ka}, who showed that the star triangle
 relations satisfied by Boltzmann weights for the Ising model imply these for a suitably defined $\tau$-function
on the $\Zb^3$ integer lattice. They were studied subsequently by Schief and others \cite{Sch, BobSch1, BobSch2, FN},
as discrete analogs of the CKP hierarchy.

Kenyon and Pemantle \cite{KePe1, KePe2} extended these to a larger  system,
which they called the {\em hexahedron recurrence}, and applied them to the study of double dimer covers
and rhombus tilings. These can either be derived  directly or,  if we include both symmetric and
nonsymmetric Pl\"ucker coordinates, by again combining the Pl\"ucker relations
with the linear Lagrangian condition.

 To see this,  multiply the short Pl\"ucker relations  (\ref{short1*}), by $S_1$ to get:
\be\label{short4}
  S_1(S_0S_{1^*} +  T_1^2 -  S_2S_3)=0.
\ee
and another short Pl\"ucker relation by $T_1$ to get
\be
\label{short3}
T_1 (S_0 T_{1^*} +  S_1T_1 - T_2T_3 )  =0.
\ee
Taking the difference gives
\begin{subequations}
\be
 S_0 (T_1 T_{1^*} - S_1S_{1^*}) =    T_1T_2T_3 -S_1 S_2 S_3 .
 \label{hexahedron1}
\ee
which, up to some changes of notation\footnote{To compare with the notation of \cite{HoSt}, \cite{KePe1} and \cite{KePe2}, set
\bea
S_0 &\&=A_0= h=a_0, \ S_1=A_1 = h_{(1)} = a_7, \ S_2 =A_2 =- h_{(2)}=-a_8,\ S_3 = A_3= h_{(3)}=a_9, \cr
S_{0^*} &\& = A_{123}= h_{(123)}  = a_0^*, \ S_{1^*} =A_{23}=h_{(23)} = a_4, \  S_{2^*}=A_{13} =-h_{(13)} = -a_5 , \  S_{3^*} =A_{12} =h_{(12)} = a_6, \cr
 T_1 &\& = h^{(x)} =a_1, \ T_2= h^{(y)}=a_2, \ T_3 = h^{(z)} = a_3, \
T_{1^*} = h^{(x)}_{(1)} = a_1^*, \ T_{2^*}  = h^{(y)}_{(2)} = a_2^*, \  T_{3^*} = h^{(x)}_{(3)} = a_3^*. \cr
 \nonumber
\eea},  is one of the hexahedron relations. The others
\bea
 S_0 (T_2 T_{2^*}  - S_2S_{2^*})&\&=    T_1T_2T_3 - S_1 S_2 S_3,
 \label{hexahedron2}  \\
 S_0 (T_3 T_{3^*}  - S_3S_{3^*}) &\&=  T_1T_2T_3 - S_1 S_2 S_3 ,
  \label{hexahedron3}  \\
S_{0^*} (T_1 T_{1^*}  - S_1S_{1^*}) &\& =  T_{1^*}T_{2^*}T_{3^*} -S_{1^*} S_{2^*} S_{3^*}  ,
 \label{hexahedron1*}  \\
S_{0^*} (T_2 T_{2^*} -  S_2S_{2^*}) &\& = T_{1^*}T_{2^*}T_{3^*} - S_{1^*} S_{2^*} S_{3^*} ,
 \label{hexahedron2*}  \\
S_{0^*} (T_3 T_{3^*} - S_2S_{2^*}) &\& =    T_{1^*}T_{2^*}T_{3^*} - S_{1^*} S_{2^*} S_{3^*},
 \label{hexahedron3*}
\eea
\end{subequations}
are derived similarly. The degree six relation (1.4d) in \cite{KePe2} follows by solving  eqs.~(\ref{hexahedron1}) - (\ref{hexahedron3})
for $T_{1^*}, T_{2^*}$ and $T_{3^*}$ and substituting either in (\ref{hexahedron1*}), (\ref{hexahedron2*}) or (\ref{hexahedron3*}).


\section{The CKP hierarchy, infinite Lagrangian Grassmannians and hyperdeterminantal relations}
\label{CKP_hierarchy}

\subsection{Baker function, Lax operators and CKP reduction}
\label{KP_Baker_CKP}

\quad We  recall the formulation of the KP hierarchy  as an infinite abelian group action on an infinite dimensional
Grassmannian \cite{Sa, SW}, its relation to isospectral flows of pseudo-differential operators
and reduction to the CKP hierarchy  \cite{DJKM1, JM1}.

    Let $\HH$ denote a separable Hilbert space, with orthonormal
basis $\{e_i\}_{i\in \Zb}$ labelled by the integers. Concretely, we may think of $\HH $ as
the space of square integrable functions  $L^2(S^1)$ on the unit circle
$S^1=\{z := e^{i\theta}, \ 0\le \theta < 2\pi\}$ in the complex plane with hermitian inner product
\be
(f, g) := \frac{1}{2\pi i} \oint_{z\in S^1} \overline{f(z) }g(z) \frac{dz}{z}.
\ee
and (for reasons of historical conventions), choose the basis elements as the monomials
\be
e_i := z^{-i-1}, \quad i \in \Zb.
\ee

Split $\HH$ as a direct sum
\be
\HH = \HH_+ \oplus \HH_-
\label{spli_hilbert_space}
\ee
of  Hardy spaces
\be
\HH_+ := \span\{z^i =e_{-i-1}\}_{i\in \Nb}, \quad \HH_- := \span\{z^{-i} =e_{i-1}\}_{i\in \Nb^+},
\label{hardy_spaces}
\ee
consisting of elements $f\in \HH_+$ that admit analytic continuation to the interior of $S^1$
and  $f\in \HH_-$ that admit analytic continuation outside $S^1$, with $f(\infty)=0$
(or, equivalently, the positive  and negative power Fourier series).
By the infinite Grassmannian $\Gr_{\HH_+}(\HH)$, we mean
a suitably defined Banach manifold (see \cite{SW}) consisting of subspaces $w\ss\HH$ that
are {\em commensurable} with $\HH_+$, in the sense that  orthogonal projection
$\pi_+: w \ra \HH_+$ is a Fredholm operator (with index $n \in \Zb$) while the
projection $\pi_-: w \ra \HH_-$ is ``small'' (either Hilbert-Schmidt, or compact,
depending on the context).

We skip the analytic details (see \cite{SW} or \cite{HB},  Chapt. 3),
and just require that, via a suitable choice of ``admissible basis", we may identify
the spaces $w$ and $\HH_+$ as isomorphic, so  it is meaningful to define
the {\em determinant} $\det(\pi_+:w \ra \HH_+)$ of the projection map.
We also define (as in \cite{SW}) the general linear group $\Gl(\HH)$ of
invertible endomorphisms of $\HH$ (satisfying certain admissibility conditions),
its Lie algebra $\grgl(\HH)$, and the abelian subgroup of  {\em shift flows} $\Gamma_+ \ss \Gl(\HH)$
\be
\Gamma_+:= \{ \gamma_+(\tb) \in \Gl(\HH),  \ \gamma_+(\tb) \gamma_+(\sb) = \gamma_+(\tb+\sb)\},
\ee
where $\tb = (t_1, t_2, \dots )$ are the KP flow variables, and the abelian group
\be
\Gamma_+:= \{\gamma_+(\tb) := e^{\xi(z,\tb)}, \quad \xi(z, \tb) := \sum_{i=1}^\infty t_i z^i\}
\ee
acts on $f\in \HH=L^2(S^1)$ by multiplication. This lifts in the standard way to an action on the Grassmannian
\bea
\Gamma_+ \times \Gr_{\HH_+}(\HH) &\&\ra \Gr_{\HH_+}(\HH) \cr
(\gamma_+(\tb), w) &\& \mapsto w(\tb) := \gamma_+(\tb) w.
\eea
The orbit of an element $w\in \Gr_{\HH_+}(\HH)$ under this action is denoted
$ \OO_w =\{w(\tb)\}$.
The KP $\tau$-function $\tau^{KP}_w$ corresponding to the element $w$ is defined to be
\be
\tau^{KP}_w(\tb) := \det(\pi_+: w(\tb) \ra \HH_+).
\ee
This then satisfies the Hirota bilinear residue equations
\be
\res_{z=\infty} (e^{\xi(z,\ {\bf \delta t})}  \tau^{KP}_w(\tb - [z^{-1}])
\tau^{KP}_w(\tb + {\bf \delta t}+  [z^{-1}]) )dz
\label{KP_hirota_res_eq}
\ee
identically in the parameters ${\bf \delta t}=(\delta t_1, \delta t_2, \dots, )$,
where
\be
[z^{-1}] := \left(\frac{1}{z}, \frac{1}{2z^2}, \dots \right),
\ee
and the formal residue $\res_{z=\infty}( \cdots)\,dz$ signifies evaluation of the coefficient of the $\frac{1}{z}$ term in the formal
Laurent series appearing in each coefficient  of the monomials in the shift parameters $\{\delta t_i\}$.

The formal Baker-Akhiezer function (or {\em wave function})  and its dual are given by the Sato formulae \cite{Sa, JM1} as
\bea
\Psi(z, \tb) &\&:= e^{\xi(z, \tb)} \frac{\tau(\tb - [z^{-1}])}{\tau(\tb)} =:  e^{\xi(z, \tb)} (1 +\sum_{i=1}^\infty{a_i}(\tb) z^{-i}),\\
\Psi^*(z, \tb) &\&:= e^{-\xi(z, \tb)} \frac{\tau(\tb + [z^{-1}])}{\tau(\tb)} =:  e^{-\xi(z, \tb)} (1 +\sum_{i=1}^\infty{a^*_i}(\tb) z^{-i}).
\label{sato_formula}
\eea
The formal pseudo-differential ``wave operator'' and its dual are defined by
\bea
\hat{W}  &\&:= 1  +\sum_{i=1}^\infty{a_i(\tb)} \partial^{-i},\\
\hat{W}^\dag  &\&:= 1  +\sum_{i=1}^\infty{a^*_i(\tb)} \partial^{-i},
\label{wave_op}
\eea
where
\be
\partial :=\frac{\partial}{\partial x},
\ee
so that
\be
\Psi(z, \tb) = \hat{W} ( e^{\xi(z, \tb)} ) , \quad  \Psi^*(z, \tb) = (\hat{W}^\dag)^{-1} ( e^{-\xi(z, \tb)} ) ,
\ee
and the Lax pseudo-differential operator operator is
\be
\LL := \hat{W} \partial \hat{W}^{-1} = \partial + \sum_{i=1}^\infty u_i (\tb) \partial^{-i}.
\label{lax_dressing}
\ee

It follows (see \cite{Sa, SW}, or \cite{HB}, Chapt. 3) that $\Psi(z, \tb)$ satisfies
\be
\frac{\partial \Psi}{\partial t_i} = \DD_i \Psi,  \quad \forall i \in \Nb,
\ee
where
\be
\DD_i := (\LL^i)_+
\ee
is the differential operator part of the pseudo-differential operator $\LL^i$,
and $\LL$ satisfies the Lax equations
\be
\frac{\partial \LL}{\partial t_i} = [ \DD_i, \LL].
\ee
The compatibility conditions
\be
\frac{\partial \DD_i}{\partial t_j} -  \frac{\partial \DD_j}{\partial t_i} + [\DD_i, \DD_j] =0
\label{KP_hierarchy}
\ee
give an infinite set of constant coefficient partial differential equations for
the functions $\{u_i(\tb)\}_{i\in \Nb^+}$,  each involving derivatives with respect to
a triple $(x, t_i, t_j)_{1< i<j}$, with $x$ identified, within a translation constant, with the flow variable $t_1$.
The functions $\{u_i(\tb)\}_{i\in \Nb^+}$ are uniquely determined,
through eqs.~(\ref{sato_formula}), (\ref{wave_op}), (\ref{lax_dressing}), in terms of derivatives of the
$\tau$-function, and the set of equations (\ref{KP_hierarchy}) are equivalent to the Hirota
residue equation (\ref{KP_hirota_res_eq}).

   The KP hierarchy is reduced to the CKP one \cite{DJKM1, JM1}  by imposing additional conditions.
   In terms of the Lax operator $\LL$,  we require the formal anti-self-adjointness
   condition
  \be
  \LL^\dag = - \LL,
 \ee
 to be satisfied, which implies that
\be
\DD^\dag_{2j -1} = - \DD_{2j-1} , \quad j \in \Nb^+.
\ee
It  follows that
\be
\tau(\tb) = \tau(\tilde{\tb}),
\ee
where
\be
\tilde{\tb} := (t_1, - t_2, t_3, -t_4, \dots),
\ee
and
\be
\Psi^*(z, \tb) = \Psi(-z, \tilde{\tb}).
\ee
The Hirota bilinear equation (\ref{hirota_bilinear_tau_res}) therefore reduces to
\be
\res_{z=0} \Psi(z, \tb) \Psi(-z, \tilde{\tb}+\delta\tilde{\tb})dz  = 0
\label{hirota_bilinear_tau_res_CKP_bis}
\ee
or, for vanishing even flow variables
\be
\res_{z=0} \Psi(z, \tb') \Psi(-z, \tb' +\delta\tb')dz  = 0,
\label{hirota_bilinear_tau_res_CKP_t_o}
\ee
which is (\ref{hirota_bilinear_tau_res_CKP}) with
\be
\Psi(z, \tb')  := \Psi_{w^0}(z, {\bf t}_o)
\ee


\subsection{Fermionic representation of KP $\tau$-functions}
\label{fermionic_KP-tau}

The fermionic Fock space is  the semi-infinite wedge product space of $\HH$ with itself
\be
\FF:= \Lambda^{\infty/2}(\HH) = \bigoplus_{n\in \Zb} \FF_n,
\ee
which is the orthogonal direct sum of the subspaces $\FF_n$ with fermionic charge $n$.
Orthonormal bases $\{|\lambda; n\rangle\}$, labelled by pairs
$(\lambda, n)$ of integer partitions $\lambda$ of any weight and integers $n\in \Zb$,
are defined as
\be
|\lambda; n\rangle := e_{l_1} \swedge e_{l_2} \swedge e_{l_3} \swedge \cdots,
\label{eq:LambdanState}
\ee
where $l_1 > l_2  > \cdots$ is a strictly decreasing sequence of integers, called {\em particle locations}
which saturates, after $\ell(\lambda)$ terms, to become a sequence of  successive decreasing integers.
These are determined in terms of the parts $\{\lambda_i\}_{i\in \Nb^+}$ of $\lambda$ and $n$
(where $\lambda_i := 0$ if   $i> \ell(\lambda)$) by
\be
l_i := \lambda_i - i + n.\quad i \in \Nb^+.
\label{particle_locations}
\ee

The vacuum element in the sector $\FF_n$ is
\be
|\emptyset; n\rangle =: |n \rangle = e_{n-1} \swedge e_{n-2} \swedge \cdots.
\ee

 As in the finite dimensional case, the Fermi creation and annihilation operators are
 elements of the representation
 \be
 \hat{\Gamma} : \Cl(\HH\oplus\HH^*, Q) \ra \End(\FF)
 \ee
 of the infinite dimensional Clifford algebra $\Cl(\HH\oplus\HH^*, Q)$,
where $Q$ is the canonical quadratic form on $\HH\oplus \HH^*$.
\be
Q(v+ \mu) = 2\mu(v), \quad v \in \HH, \ \mu \in \HH^*,
\ee
generated by exterior and interior multiplication by the basis elements and their duals:
\bea
\hat{\Gamma}_{X +\xi} &\&:= X \wedge + i_{\xi} \in \End(\FF),   \quad X\in \HH, \ \xi \in \HH^*,
\\
\hat{\Gamma}_{e_i} &\&= \psi_i := e_i \swedge, \quad  \hat{\Gamma}_{e^*_i} = \psi_i^\dag := i_{e_i^*}.
\eea
These satisfy the anticommutation relations
\be
[\psi_i, \psi_j ]_+ =  [\psi^\dag_i, \psi^\dag_j ]_+ =0, \quad [\psi_i, \psi^\dag_j ]_+= \delta_{ij},
\ee
and the vacuum annihilation conditions
\be
\psi_{-i}|0 \rangle = 0, \quad \psi^\dag_{i-1}|0 \rangle =0 , \ \forall \  i \ \in \Nb^+.
\ee
An equivalent way \cite{HB} of representing the basis elements is then
\be
|\lambda ; n\rangle = (-1)^{\sum_{i=1}^r b_i }\prod_{i=1}^r \psi_{a_i+n} \psi^\dag_{-b_i -1 +n} | n \rangle,
\ee
where $(a_1, \dots, a_r | b_1 , \dots, b_r)$ are the Frobenius indices of the partition $\lambda$.

The Clifford representation  of elements of the Lie algebra $\grg\grl(\HH)$ is
\be
\hat{A} = \sum_{i, j \in \Zb }A_{ij} \no{\psi_i \psi^\dag_j},
\ee
where normal ordering $\no{\OO}$ of bilinear elements means
\be
\no{\psi_i\psi^\dag_j}:= \psi_i\psi_j^\dag - \langle 0 | \psi_i \psi^\dag_j| 0 \rangle.
\ee
The corresponding Clifford representation of an element $g = e^A \in \Gl_0(\HH)$  in the
identity component of the general linear group $\Gl(\HH)$ is given by exponentiation
\be
\hat{g}= e^{\hat{A}}.
\ee
The current components are defined by
\be
J_i := \sum_{j\in \Zb}\psi_j \psi^\dag_{j+i}, \quad \text{for} \ \pm i \in \Nb^+ ,
\ee
and the abelian group of KP ``shift flows''  is represented fermionically by
\be
\hat{\gamma}_+(\tb) = e^{\sum_{i=1}^\infty t_i J_i}.
\ee

For any $g\in \Gl_0(\HH)$ in the identity component of  the general linear group $\Gl(\HH)$) for which
\be
w := g(\HH_+) \in \Gr^0_{\HH_+}(\HH)
\ee
belongs to the virtual dimension $0$  component   $\Gr^0_{\HH_+}(\HH)$ of the Segal--Wilson Grassmannian \cite{SW, HB}  ),
the corresponding KP $\tau$-function is given  by the fermionic  vacuum expectation value (see \cite{DJKM3},  or \cite{HB}, Chapt. 5)
\be
\tau^{KP}_w(\tb) = \langle 0 | \hat{\gamma}_+(\tb) \hat{g} | 0 \rangle,
\label{tau_w_fermionic_VEV}
\ee
where $\hat{g}$ is the fermionic representation of $g$.

Under the bosonization isomorphism (in the zero fermionic charge sector $\FF_0$)
\be
\II_{\BB\FF}: | v ; 0 \rangle \mapsto \langle 0 | \hat{\gamma}_+(\tb) | v ; 0\rangle,
\label{bosoniz_isomorph}
\ee
the element $\hat{g}|0 \rangle$ gets mapped to the KP  $\tau$-function $\tau^{KP}_w(\tb)$,
the basis elements $|\lambda; 0\rangle$ get mapped to Schur functions
\be
\II_{\BB\FF}( |\lambda; 0\rangle) = s_\lambda(\tb)
\label{bosoniz_schur}
\ee
and the current components get mapped to:
\be
\II_{\BB\FF} \cdot J_i \cdot \II^{-1}_{\BB\FF}= i \frac{\partial}{\partial t_i},\quad \II_{\BB\FF}\cdot J_{-i} \cdot \II^{-1}_{\BB\FF}= t_i.
\label{bosoniz_J_i}
\ee
The Pl\"ucker coordinates $\{\pi_\lambda(w)\}$ of the element $w = g(\HH_+)$
appearing as coefficients in the expansion  over Schur functions
\be
\tau_w^{KP}(\tb) = \sum_\lambda \pi_\lambda(w)  s_\lambda(\tb)
\label{KP_tau_schur_exp}
\ee
of  the  $\tau$-function $\tau_w^{KP}(\tb)$ defined in (\ref{tau_w_fermionic_VEV}) 
are the fermionic matrix elements
\be
\pi_\lambda(w) := \langle \lambda;0 | \hat{g}|0\rangle,
\ee
and the Pl\"ucker map $\Pl_{\HH_+}: \Gr_{\HH_+}(\HH) \ra \Pb(\FF)$  applied to
an element $w$ with admissible basis  (\cite{SW})  $\{w_1, w_2, \cdots\}$ gives
\be
\Pl_{\HH_+}(w):=  [w_1 \swedge w_2 \swedge \cdots ]  =:[|w\rangle] = \big[ \sum_\lambda   \pi_\lambda (w) |\lambda; 0\rangle\big]
 = \bigcap_{i\in \Nb} \ker(\hat{\Gamma}_{w_i}).
 \label{inf_dim_plucker_map}
\ee


\subsection{Symplectic form $\omega$ on $\HH$, Lagrangian Grassmannians and fermionic representation of $\symp(\HH, \omega)$ }
\label{symplectic_form_HH}

We define the symplectic form $\omega$ on $\HH = L^2(S^1)$  by
\be
\omega(f,g) = \frac{1}{2\pi i}\oint_{z\in S^1} f(z) g(-z) dz.
\ee
The subspaces $\HH_\pm \ss \HH$ are maximal isotropic, i.e. Lagrangian, with respect to
$\omega$, and may be viewed as mutually dual  $\HH_- \sim \HH_+^*$ under the pairing
\be
f_-(g_+) := \omega(f_-, g_+) \quad \text{for } f_- \in \HH_-\ g_+ \in \HH_+.
\ee
In terms of this pairing, the symplectic form is
\be
\omega( f_- + f_+, g_- + g_+) = f_-(g_+) - g_-(f_+), \quad \text{for } \ f_\pm, g_\pm \in \HH_\pm
\ee
or, in terms of basis elements
\be
\omega(e_i, e_j) =-\omega(e_j, e_i)= (-1)^i \delta_{i, -j -1}, \quad e_{i}(e_{-j-1} )= (-1)^i \delta_{ij}, \quad i, j \in \Zb.
\ee

Following \cite{JM1}, the fermionic representation of the $C_\infty$ Lie algebra is realized as the
subalgebra of $\grg\grl(\HH) \sim A_\infty$ consisting of the fixed points
\be
\sigma_{-1}(\hat{A}) = \hat{A}
\ee
under the Clifford algebra automorphism generated by
\be
\sigma_{-1} (\psi_i) := (-1)^{i+1} \psi^\dag_{-i-1}, \quad \sigma_{-1} (\psi^\dag_i) := (-1)^{i+1} \psi_{-i-1}.
\label{sigma_def}
\ee

The entire algebra is  generated by forming successive commutators  from the Chevalley basis elements:
\bea
\hat{E}_0 &\&= \psi_{-1}\psi^\dag_0, \quad \hat{F}_0 =\psi_{0}\psi^\dag_{-1},
 \quad \hat{H}_0 = \psi_{-1}\psi^\dag_{-1}- \psi_0 \psi^\dag_0, \cr
 &\& \cr
 \hat{E}_j &\&= \psi_{j-1} \psi^\dag_j + \psi_{-j-1}\psi^\dag_{-j},  \quad
\hat{F}_j = \psi_{j} \psi^\dag_{j-1} + \psi_{-j}\psi^\dag_{-j-1}, \quad \text{for }  j\ge 1,\cr
 &\& \cr
 \hat{H}_j &\& = \psi_{j-1} \psi^\dag_{j-1}  - \psi_{j} \psi^\dag_{j} + \psi_{-j-1}\psi^\dag_{-j-1}- \psi_{-j}\psi^\dag_{-j},
 \quad  \text{for }  j\ge 1.
 \label{chevalley}
\eea
This corresponds to the following representation on $\HH$ as  generators of a subalgebra of $A_\infty \sim \grgl(\HH)$:
\bea
E_0 e_i &\&= \delta_{i,0} \, e_{-1}, \quad F_0 e_i = \delta_{i, -1}\, e_0,
\quad H_0 e _i = \delta_{i,-1} \, e_{-1} - \delta_{i, 0} \, e_0, \cr
&\& \cr
E_j e_i &\&= \delta_{i,j} \, e_{j-1} + \delta_{i, -j} \, e_{-j-1}, \quad F_j \, e_i = \delta_{i, j-1} \, e_j + \delta_{i, -j-1} \,e_{-j}, \quad   \text{for } j \ge 1\cr
&\& \cr
H_j e_i &\&=(\delta_{i, j-1} - \delta_{i,j} + \delta_{i, -j-1} - \delta_{i, -j}) \, e_i,\quad   \text{for } j \ge 1.
\eea
It follows that all these elements $X$ satisfy
\be
\omega(X e_i, e_j) + \omega(e_i, X e_j) =0,
\label{sympl_invar_chevalley}
\ee
and so do the commutators $[X, Y]$ of any two such elements, and all successive
commutators, and hence any element $X \in C_\infty$.
The symplectic form $\omega$ is therefore  invariant under this $C_\infty$ action, and we may
 identify $C_\infty \sim \symp(\HH, \omega) \ss \grgl(\HH)$.

The orbit of $\HH_+\ss \HH$ under the subgroup $\Sp(\HH, \omega) \ss \Gl(\infty)$ preserving the
symplectic form $\omega$ is the Lagrangian Grassmannian $\Gr^\LL_{\HH_+}(\HH, \omega) \ss \Gr_{\HH_+} (\HH)$
consisting of maximal isotropic subspaces $w^0\ss \HH$,  on which the restriction of $\omega$ to any
$w^0\in  \Gr^\LL_{\HH_+}(\HH, \omega)$ vanishes
\be
\omega \vert_{w^0} =0.
\label{inf_lagrangian_cond}
\ee

Bose-Fermi equivalence identifies the basis state $|\lambda\rangle $
in the zero fermionic charge sector $\FF_0$  with the Schur function $s_\lambda(\tb)$. The operators
that are the fermionization of the (Murnaghan-Nakayama) operator \cite{St} of multiplication by $j t_j$ and its dual $\frac{\partial}{\partial t_j}$
when acting on the basis of Schur functions are the current components
\bea
J_{-j} &\&:= \sum_{i\in \Zb} \psi_i\psi^\dag_{i-j} = \II^{-1}_{\FF \BB} \cdot j t_j \cdot \II_{\FF \BB} , \cr
J_j &\&:= \sum_{i\in \Zb} \psi_{i}\psi^\dag_{i+j} =  \II^{-1}_{\FF \BB} \cdot  \frac{\partial}{\partial t_j} \cdot \II_{\FF \BB} , \quad r\in \Nb^+ .
\label{fermi_murnaghan_op_dual}
\eea


\subsection{Decomposition of $\FF$ into $\Sp(\HH, \omega)$ invariant submodules}
\label{infin_Sp_decomp}

The operator $i_{\omega_N}$ defined in (\ref{eq:OperatorOmegaNDag}), may be identified in the infinite dimensional setting as  the fermionic operator
\be
\hat{\omega}^\dag := \sum_{i=0}^\infty (-1)^i \psi^\dag_{-i-1}\psi^\dag_{i},
\label{hat_omega_def}
 \ee
which lowers the fermionic charge by $2$.
Denote the kernel of $\hat{\omega}^\dag $, restricted to $\FF_0$, as
\be
\FF^{(0)}_0 := \{|v \rangle \in \FF_0 \ | \ \hat{\omega}^\dag | v \rangle= 0\}.
\ee
This is the infinite dimensional counterpart of the $\Sp(W, \omega_N)$ invariant
submodule $P^N_N \ss \Lambda^N(W)$ defined in Section \ref{decomp_Sp_2N_irreps}.
We also define the dual fermionic operator
\be
\hat{\omega} := -\sum_{i=0}^\infty (-1)^i \psi_{-i-1}\psi_{i},
\label{eq:OmegaHatDefinition}
\ee
which raises the fermionic charge by $2$, and has the same kernel
\be
\FF^{(0)}_0 := \{|v \rangle \in \FF_0 \ | \ \hat{\omega} | v \rangle= 0\}.
\ee
It follows that both $\hat{\omega}$ and $\hat{\omega}^\dag$ commute with all elements of $\symp(\HH, \omega)$.
\begin{lemma}
\label{omega_commute}
\be
[\hat{\omega}, \hat{X}]  = 0 \quad \text{and} \quad [\hat{\omega}^\dag, \hat{X}]  = 0 \quad \forall \ \hat{X} \in \symp(\HH, \omega).
\ee
\end{lemma}
\begin{proof} This is a direct computation for the case of the Chevalley elements (\ref{chevalley}).
By the Jacobi identity, it  also holds for all commutators of such elements, and hence for
all elements $\hat{X} \in \symp(\HH,\omega)$.
\end{proof}
\begin{remark}
Note that the automorphism $\sigma_{-1}$ in (\ref{sigma_def}) may be expressed as
\be
\sigma_{-1}(\psi_i) = [\hat{\omega}, \psi_i], \quad  \sigma_{-1}(\psi^\dag_i) = [\hat{\omega}^\dag, \psi^\dag_i].
\ee
\end{remark}
\begin{definition}
For all  $j, n \in \Nb$, define the subspaces
\be
\FF^{(j)}_n :=( \hat{\omega})^j \left( \ker \left( \hat{\omega}^\dag \vert_{\FF_{n-2j}}\right) \right)\ss \FF_n.
\label{eq:FnJDefinition}
\ee
\end{definition}

As in the finite dimensional case, we have a direct sum decomposition:
\begin{proposition}(cf.~\cite{DJKM3, JM1, JM2})
\label{prop:FF_decomp}
The fermionic Fock space decomposes into a direct sum of $\mathfrak{sp}(\HH,\omega)$ submodules :
\be
\FF =\bigoplus_{n\in\Zb} \bigoplus_{j\in\Nb^+} \FF^{(j)}_n
\label{eq:FnjSpModuleDecomposition}
\ee
For any pair of integers $n\in\Zb$ and $j\in\Nb$ satisfying $n\leq j$,  $\mathcal F_n^{(j)}$ is irreducible.
\label{prop:FmjIrreducible}
\end{proposition}

In what follows, we only consider the submodules $\FF^{(j)}_0$ that lie within the zero fermionic charge sector $\FF_0$.
As in the finite dimensional case, these are all highest weight modules. To see this, define, for all  $j \in \Nb^+$,
the element
\be
|v(j)\rangle := (\hat{\omega})^j   | \hspace{-2 pt} -\hspace{-2 pt} 2j \rangle,
\ee
i.e., the image of the vacuum element in the fermionic charge  sector  $\FF_{-2j}$
under the $j$th power of the symplectic ``raising'' map $\hat{\omega}$.
Since, as is easily verified, the  charged vacuum vector $|\hspace{-2 pt}- \hspace{-2 pt}2j \rangle \in \FF_{-2j}$
is in the kernel of $\hat{\omega}^\dag \vert_{\FF_{-2j}}$
\be
\hat{\omega}^\dag |\hspace{-2 pt}-\hspace{-2 pt}2j\rangle =0,
\ee
it follows that $|v(j)\rangle \in \FF^{(j)}_0$.
Lemma \ref{omega_commute} then implies that the $\symp(\HH, \omega)$ action on  $|v(j)\rangle $
is the same as its action on the vacuum vector $|\hspace{-2 pt}-\hspace{-2 pt} 2j\rangle \in \FF_{-2j}$ in each sector.
We have
\be
|\hspace{-2 pt}-\hspace{-2 pt}2j \rangle = e_{\hspace{-2 pt}-\hspace{-2 pt}2j-1} \swedge e_{-2j-2} \swedge \cdots,
\ee
so the action of the raising operators $\{\hat{E}_m\}_{m\in \Nb}$,  and  the Cartan elements $\{\hat{H}_{m\in \Nb}\}$
 on the vacuum vector $|\hspace{-2 pt}- \hspace{-2 pt} 2j\rangle$ and hence also on $|v(j)\rangle$ are easily computed.
\begin{lemma}
\label{chevalley_highest_weight}
\be
\hat{E}_m |v(j)\rangle = 0,  \quad  \hat{H}_m |v(j)\rangle = \delta_{j,m} |v(j)\rangle, \quad  \forall \ m \in \Nb.
\ee
\end{lemma}
\begin{proof}
The corresponding relations  on $|\hspace{-2 pt}-\hspace{-2 pt} 2j\rangle$,  are verified directly from the definition of the operators $\{\psi_i, \psi^\dag_i\}$.
They therefore also hold on $|v(j)\rangle$ by the equivariance of the maps $(\hat{\omega}^\dag)^j$  and $(\hat{\omega})^j$ implied by
Lemma  \ref{omega_commute}.
\end{proof}
Since the $|v(j)\rangle$'s are  annihilated by the raising operators $\{\hat{E}_m\}_{m\in \Nb}$ and are eigenvectors of the Cartan elements
  $\{\hat{H}_m\}$,  with  weights given by the eigenvalues $\{\delta_{m,j}\}$,
   they are highest weight vectors in the various submodules $\FF^{(j)}_0$.
The submodules $\FF^{(j)}_0$ can therefore be viewed as the linear span of
the elements $ \{\hat{F}_{\alpha}|v(j)\rangle\}$, where $\alpha = (\alpha_1 > \cdots > \alpha_r =2j)$
is any strict partition ending with $\alpha_r =2j$, for  all $r\in \Nb$ and
\be
\hat{F}_\alpha := \hat{F}_{\alpha_1}  \cdots  \hat{F}_{\alpha_r}.
\ee

\begin{remark}
The lowering operators $\{\hat{F}_m\}$ for $m\neq 2j$ also annihilate the highest weight vectors $| v(j)\rangle$:
\be
\hat{F}_m |v(j)\rangle = 0 \quad \text{if} \ m\neq 2j.
\ee
The elements $\hat{F}_\alpha$ may be viewed as spanning the universal enveloping
algebra $U(\NN_-)$  of the subalgebra $\NN_- \ss \symp(\HH, \omega)$ generated
by the Chevalley elements $\{\hat{F}_m\}_{m\in \Nb}$.
The submodule $\FF^{(j)}_0$ may be viewed as a quotient of the Verma module corresponding to this universal
enveloping algebra, with highest weight the same as $|v(j)\rangle$,
 for which we choose a basis $\{\FF_\alpha\}$ labelled by the strict partitions corresponding to the elements
 $\{\hat{F}_\alpha\}$, and quotient by the span of all those elements $\FF_\alpha$ for which $| v(j)\rangle$
 is in the kernel of $\hat{F}_\alpha$.
\end{remark}

It follows, as in the finite dimensional case (Prop.~\ref{plucker_lagrangian_linear_image}),
that the image of  the Lagrangian Grassmannian $\Gr^\LL_{\HH_+}(\HH, \omega) \ss \Gr_{\HH_+}(\HH)$
under the Pl\"ucker map (\ref{inf_dim_plucker_map}) is contained within the kernel $\FF^{(0)}_0$ of $\hat{\omega}$
(or $\hat{\omega}^\dag$) acting on $\FF_0$.
For $w^0 \in \Gr^\LL_{\HH_+}(\HH, \omega) $, let
\be
[|w^0\rangle] := \Pl_{\HH_+}(w^0)
\ee
denote its image under the Pl\"ucker map. Then 
\be
\hat{\omega}  | w^0\rangle =0,  \quad \hat{\omega}^\dag  | w^0\rangle =0, \quad\forall \ w^0\in \Gr^\LL_{\HH_+}(\HH, \omega) ,
\label{null_lagrangian_cond}
\ee
and that these kernels are equal to the entire submodule $\FF^{(0)}_0$
\be
\ker(\hat{\omega}) |_{\FF_0} = \ker(\hat{\omega}^\dag) |_{\FF_0}  = \FF^{(0)}_0.
\ee
\begin{proposition}
\label{image_lagrangian_inf}
The images $\{[|w^0\rangle] \}$ of the elements $w^0  \in \Gr^\LL({\HH_+}(\HH, \omega)$ of  the Lagrangian Grassmannian
under the Pl\"ucker map span the $\symp(\HH, \omega)$-submodule $\FF^{(0)}_0$.
\end{proposition}
\begin{proof}
By construction, the image of the Lagrangian Grassmannian must span a nontrivial $\mathfrak{sp}(\mathcal H,\omega)$ submodule of $\FF_0^{(0)}$. By Proposition \ref{prop:FmjIrreducible} we know that $\FF_0^{(0)}$ is irreducible, and hence must coincide with the span of the image of the Lagrangian Grassmannian.
\end{proof}


\subsection{The CKP reduction condition}
\label{CKP_reduction}

Combining these results, it follows that a KP $\tau$-function admitting a Schur function expansion
\be
\tau^{KP}(\tb) = \sum_{\lambda} \pi_\lambda s_\lambda(\tb)
\ee
 is of CKP type if  and only if its fermionic counterpart $\sum_\lambda \pi_\lambda | \lambda \rangle$
 is in the submodule $\FF^{(0)}_0 \ss \FF_0$; i.e. if, in addition to the Pl\"ucker relations,
 the linear constraint
 \begin{subequations}
 \be
 \hat{\omega} (\sum_\lambda \pi_\lambda |\lambda\rangle) = 0
 \label{null_fermi_omega}
 \ee
 or, equivalently,
 \be
 \hat{\omega}^\dag (\sum_\lambda \pi_\lambda |\lambda\rangle) = 0
 \label{null_fermi_omega_dual}
 \ee
  \end{subequations}
is satisfied.
This may be expressed equivalently as a set of linear relations for the Pl\"ucker coefficients.

Another way to express the fact that a KP $\tau$-function
\be
\tau^{KP}_{w^0}, \quad w^0 = h(\HH_+)
\ee
 is the bosonization of an element in $\FF^{(0)}_0$ is to note that in the fermionic VEV  representation
\be
\tau^{KP}(\tb) = \langle 0 | \hat{\gamma}_+(\tb) \hat{h}| 0\rangle,
\ee
the group element $h$ belongs to $\Sp(\HH, \omega)$,  so that
\be
\sigma_{-1}(\hat{h}) = \hat{h}.
\label{h_sigma_1_inv}
\ee
From its definition (\ref{sigma_def}), $\sigma_{-1}$ acts on the shift flow current component generators as
\be
\sigma_{-1}(J_j) = (-1)^{j+1} J_j,
\ee
and therefore
\be
\sigma_{-1}\left(\hat{\gamma}_+(\tb)\right)  =\hat{\gamma}_+ (\tilde{\tb}),
\ee
where
\be
\tilde{\tb} := (t_1, - t_2, t_3, -t_4, \dots ).
\ee
Since the (right) ideal of the fermionic Clifford algebra $\Cl(\HH\oplus\HH^*), Q)$ generated by the
annihilators $\{\psi_{-i}, \psi^\dag_{i-1}\}_{i\in \Nb}$ of the vacuum $|0 \rangle$
is invariant under $\sigma_{-1}$, the action of $\sigma_{-1}$ passes to the quotient by
this ideal, and hence projects to the Fock space, so that
\be
\langle 0 | \sigma_{-1}(\OO) |0  \rangle = \langle  0| \OO|0\rangle
\ee
for any element $\OO \in \Cl(\HH+\HH^*,Q)$ of the Clifford algebra. Therefore
\be
\langle 0 | \sigma_{-1}\left(\hat{\gamma}_+(\tb) \hat{h}\right) | 0 \rangle =  \langle 0 |\left(\hat{\gamma}_+(\tilde{\tb}) \hat{h}\right) | 0 \rangle ,
\ee
and hence
\be
\tau^{KP}_{w^0}(\tb) = \tau^{KP}_{w^0}(\tilde{\tb}), \quad \forall \ \tb=(t_1, t_2, \dots).
\label{tau_KP_CKP_symm}
\ee
In particular, this implies the conditions
\be
\frac{\partial \tau^{KP}_{w^0}(\tb)}{\partial t_{2j} }\bigg|_{\tb_e = {\bf 0}}=0, \quad  \forall \ j \in \Nb^+ .
\label{KZ_2j_linear_cond}
\ee

As explained in the introduction, the square of any CKP $\tau$-function $\tau^{CKP}_{w^0}({\bf t}_o)$
can be expressed as the restriction to  $\tb':= (t_1, 0, t_3, 0, \cdots )$ of a KP $\tau$-function $\tau^{KP}_{w^0}(\tb)$,
\be
(\tau^{CKP}_{w^0}({\bf t}_o))^2 = \tau^{KP}_{w^0}(\tb'), \quad
\label{KZ_CKP_square_ KP}
\ee
satisfying the auxiliary criticality conditions \cite{KZ}.
 It follows that we have a Schur function expansion
\be
\tau^{KP}_{w^0}(\tb') = \sum_{\lambda}\pi_\lambda(w^0)s_\lambda(\tb'),
\label{square_schur_expansion}
\ee
in which the  Pl\"ucker coordinates $\{\pi_\lambda(w^0)\}$  are subject to the linear constraints
(\ref{KZ_2j_linear_cond}}).

To find these explicitly, we first recall the Murnaghan-Nakayama rule \cite{St}, which gives the product of  any
Schur function $s_\lambda$  with the power sum symmetric functions $p_r = r t_r$, $r\in \Nb^+$.
To express this concisely, let $\Xi(f)$ be the space of formal  linear combinations,
with complex coefficients, of symbols $f_\lambda$ indexed by
elements of the Young lattice of integer  partitions $\lambda$. Define
\be
M_r: \Xi(f) \ra \Xi(f)
\ee
to be the linear map generated by
\be
M_r(f_ \lambda ):= \sum_\mu (-1)^{h(\mu/\lambda)+1} f_\mu,
\label{murnaghan_map}
\ee
where the sum is over partitions $\mu$ of weight $|\lambda|+r$ obtained by augmenting
the Young diagram for $\lambda$ by adding $r$ squares, such that the skew partition $\mu/\lambda$
is a continuous border strip (i.e. of width = $1$ and  height $h(\mu/\lambda)$, and let
$M^*_r: \Xi(f) \ra \Xi(f)$ be the dual map generated by
\be
M^*_r(f_ \lambda ):= \sum_\mu (-1)^{h(\lambda/\mu) +1} f_\mu,
\label{dual_murnaghan_map}
\ee
where the sum is over partitions $\mu$ of weight $|\lambda | -r$ obtained by reducing
the Young diagram for $\lambda$ by removing $r$ squares, such that the skew partition $\lambda/\mu$
is a continuous border strip of height $h(\mu/\lambda)$.

Viewing the Schur functions $s_\lambda(\tb)$  as  weighted homogeneous polynomials
in the normalized power sums
\be
t_i := \frac{1}{i}p_i, \quad i \in \Nb,
\ee
 the Murnaghan-Nakayama rule may be expressed as
\be
p_r s_\lambda = rt_r s_\lambda =   M_r(s_\lambda), \quad r \in \Nb^+
\label{murnaghan_rule}
\ee
and the dual Murnaghan-Nakayama rule as:
\be
\frac{\partial s_\lambda}{\partial t_r} = M^*_r(s_\lambda).
\label{dual_murnaghan_rule}
\ee

 Identifying the linear space $\Xi(f)$ with $\FF_0$ such that
 \be
 f_\lambda \sim |\lambda\rangle,
 \ee
 it follows from the bosonization map that eqs.~(\ref{murnaghan_rule}), (\ref{dual_murnaghan_rule}) may
equivalently be expressed fermionically as
\be
J_{-r} | \lambda\rangle =M_r (|\lambda\rangle) ,  \quad  J_r | \lambda\rangle = M^*_r (|\lambda\rangle).
\label{fermionic_murnaghan2}
\ee
Therefore, If
\be
[|w^0 \rangle]  = \Pl_{\HH_+} (w^0)= [\sum_\lambda  \pi_\lambda(w^0)|\lambda \rangle]
\label{lambda_basis_expansion_w_0}
\ee
 is the image under the Pl\"ucker map of an element  $w^0 \in \Gr^\LL_{\HH_+} (\HH, \omega)$, we have
 \be
 J_{-r} | w^0\rangle =M_r (|w^0\rangle)  \quad  J_r | w^0\rangle = M^*_r (|w^0\rangle).
\label{fermionic_murnaghan2}
\ee

Also note that if all the even flow variables  are set equal to $0$,
\be
\tb_e:=(t_2, t_4, \dots) = (0,0 \dots),\quad \tb = \tb':=(t_1, 0, t_3, 0, \dots),
\ee
the value of the Schur function $s_\lambda(\tb')$ equals that for the transposed partition
\be
s_\lambda(\tb') = s_{\lambda^T}(\tb').
\label{schur_CKP_sym}
\ee
Defining the orthogonal projector
 \be
\Pi_S |\lambda \rangle   \mapsto \tfrac{1}{2}\left( |\lambda\rangle + |\lambda^T\rangle\right),
\label{transpose_sym_proj}
 \ee
 (extended linearly), whose image is the  subspace consisting of elements
that  are invariant under  the transpose  involution $|\lambda\rangle \ra |\lambda^T\rangle$,
the fermionic expression of the linear constraint (\ref{KZ_2j_linear_cond}) is therefore
\be
J_{2j} \circ \Pi_S |w^0\rangle  = 0, \quad \forall\ j\in \Nb^+.
\label{fermionic_KZ_2J_linear_cond}
\ee

 Dualizing, eq. (\ref{fermionic_KZ_2J_linear_cond}) can equivalently be written in terms of the Pl\"ucker coefficients in
 the expansion (\ref{lambda_basis_expansion_w_0}).
 \begin{proposition}
 \label{prop:null_vanishing_conds}
The reduction conditions  (\ref{KZ_2j_linear_cond}) are equivalent to the following
set of  linear relations satisfied by the Pl\"ucker coefficients in the expansion (\ref{square_schur_expansion}):
\be
M_{2j}(\pi_{\lambda}(w^0)) + M_{2j}(\pi_{\lambda^T}(w^0))=0, \quad \forall \ j \in \Nb^+.
\label{plucker_KZ_2j_cond}
\ee
\end{proposition}


\subsection{Lagrange map and hyperdeterminantal relations}
\label{infinite_Lagrange_hyperdeterminantal}


\subsubsection{Lagrange map}
\label{lagrange_map_inf}

As in the finite dimensional setting, we define the subspace $\FF^S_{0} \ss \FF^{(0)}_0 \ss \FF_0$
as  the span of the basis elements corresponding to symmetric partitions
\be
\FF^S_0 := \span\{| \lambda \rangle\} \ss \FF^{(0)}_0, \quad \lambda = \lambda^T.
\label{FS_0_def}
\ee
Equivalently, we may identify  $\FF^{S}_0$ with the semi-infinite wedge product space
\be
\FF^S: = \Lambda^{\infty/2} (\HH_+)
\label{FS_def}
\ee
spanned by basis vectors
\be
e_{-J^c} := (-1)^{r+ \sum_{i=1}^r J_i }\prod_{i=1}^r\psi^\dag_{-J_i}|0\rangle = e_{\ell_1} \swedge e_{\ell_2} \swedge \cdots,
\label{FS_FS_0_isom}
\ee
in which $J \ss\Nb^+$ is a subset $\{J_1,  \dots, J_r\}$ of  the positive integers of cardinality $r$,
 ordered increasingly, so the sequence of indices  $(l_1> l_2 \cdots)$ are decreasing
negative integers that eventually saturate to a sequence of consecutive negative integers.
The basis elements of $\FF^S_0$ are those,  in the fermionic sectors $\{ \FF_{-r}\}$,
that correspond to symmetric partitions
 \be
e_{-J^c}  \leftrightarrow  |\lambda; -r\rangle,
\label{isom_FF_S_FF_S0}
\ee
where, in Frobenius notation
\be
\lambda = (J_r-1, \cdots , J_1-1 | J_r-1, \cdots , J_1-1 ).
\ee
The Lagrange map
\be
\LL: \Gr^\LL_{\HH_+}(\HH, \omega)  \ra \Pb(\FF^S)
\ee
 is then defined, as in the finite dimensional case (\ref{lagrange_map_finite_basis}), by
\be
\LL(w^0) :=\big[ \sum_J  \LL_J (w^0) e_{-J^c}\big] ,
\label{lagrange_map_inf_def}
\ee
where
\be
\LL_J(w^0) := \pi_{(J_r-1, \cdots , J_1-1 | J_r-1, \cdots , J_1-1 )}.
\label{lagrange_sym_plucker}
\ee


 \subsubsection{Hyperdeterminantal relations in infinite dimensions}
 \label{inf_hyperdet_rels}

 We again extend our definition of the Lagrange coefficients $\LL_J$ to allow the multi-index  $J=(J_1, \dots, J_r)$
 to appear in arbitrary order without changing the value of $\LL_J$.
 Choose  a triplet $(j_1, j_2, j_3)$ of distinct positive integers, and an $r$-tuplet $J$ of positive integers
 that does not contain any of these. As in Section \ref{lagrange_princip_minors},  for $a, b = 1,2,3, \ a\neq b $,
 we mean by $(J, j_a)$, $(J, j_a, j_b)$ and $(J, j_1, j_2, j_3)$  the distinct  $r+1$, $r+2$ and $r+3$-tuples consisting of the indicated
 sets of indices. As in the finite dimensional case (Proposition \ref{hyperdet_finite}), we have

 \begin{proposition}
 \label{hyperdet_inf}
 The coefficients $\LL_J$ in (\ref{lagrange_map_inf_def}) satisfy the hyperdeterminantal relations
  \bea
 &\& \LL_J^2\LL^2_{J, j_1,j_2,j_3} +\LL^2_{J, j_1} \LL^2_{J, j_2, j_3}  +\LL^2_{J, j_2} \LL^2_{J, j_1, j_3} +\LL^2_{J, j_3} \LL^2_{J, j_1, j_2}  \cr
&\&-2\LL_J \LL_{J,j_1}\LL_{J, j_2, j_3}\LL_{J, j_1. j_2, j_3}  - 2\LL_J  \LL_{J,j_2}\LL_{J, j_1, j_3}\LL_{J, j_1. j_2, j_3}  -2\LL_J  \LL_{J,j_3}\LL_{J, j_1, j_2}\LL_{J, j_1. j_2, j_3} \cr
 &\& -2\LL_{J,j_1}\LL_{J, j_2}\LL_{J, j_1,j_3}\LL_{J,j_2, j_3}  -2\LL_{J, j_1}\LL_{J, j_3}\LL_{J, j_1,j_2}\LL_{J,j_2, j_3}
 - 2\LL_{J,j_2}\LL_{J, j_3}\LL_{J, j_1,j_2}\LL_{J, j_1, j_3} \cr
 &\&   - 4 \LL_{J}\LL_{J,j_1, j_2}\LL_{J, j_1, j_3} \LL_{J, j_2, j_3}  - 4 \LL_{J, j_1}\LL_{J, j_2} \LL_{J, j_3} \LL_{J, j_1, j_2, j_3} =0.
 \cr
 &\&
 \label{hyper_det_rels_inf}
 \eea
These determine the image of the Lagrange map (\ref{lagrange_map_inf_def}) on a Zariski
open set of $\Pb(\FF^S_0)$. The inverse image $\LL^{-1}(\LL(w^0))$ of any element $\LL(w^0)$
in the variety cut out by these relations is the orbit of $w^0$ under the group
$(\Zb_2)^\infty = \{ {\boldsymbol \epsilon} :=\{ \epsilon_i=\pm 1\}_{i \in \Zb} \ss \Sp(\HH)\}$,
 acting by reflections:
 \be
 {\boldsymbol \epsilon} : (e_{-i-}, e_i) \mapsto (\epsilon_ie_{-i-1}, \epsilon_i e_i)
 \ee
  within the coordinate planes $\{e_{-i-1}, e_i\}_{i \in \Nb}$.
 \end{proposition}

To prove these relations in the infinite dimensional setting recall that, in finite dimensions, the result was obtained
by first contracting and then  projecting down in a family of different ways to $\Lambda^3(\Cb^6)$,
showing that  a generic element satisfied the Pl\"ucker conditions and the linear isotropy conditions if and only if these
reductions to $(3,6)$ dimensions also satisfied the Pl\"ucker conditions and linear isotropy conditions.
 From this, it was possible to manipulate the quadratic and linear constraints, eliminating all but the symmetric
 partition Pl\"ucker coordinates, to obtain the quartic hyperdeterminantal relations for the various cases of $(3,6)$
 dimensions.  Tracing back, the remaining variables are exactly the coefficients of the image of the
 Lagrange map, and the relations are those given  in eq.~(\ref{hyper_det_rels_inf}).

In infinite dimensions we proceed essentially the same way, but first reduce to a nested sequence of elements
$\phi_N^N \in \Lambda^N(\HH_N)$ of the $N$-th exterior  power of  a nested sequence of finite dimensional symplectic subspaces $\HH_N$,
showing that if both the Pl\"ucker relations  and the isotropy condition are satisfied for $N>N_0$  for some $N_0$
 depending on the subspace, the result in infinite dimensions follows by taking the direct limit.

Thus, consider the  image of the Lagrangian Grassmannian $\Gr^\LL_{\HH_+}(\HH, \omega) \ss \Gr_{\HH_+} (\HH)$
under the Pl\"ucker map (\ref{plucker_map_def}), intersected with the kernel $\FF^{(0)}_0$ of $\hat{\omega}$
(or $\hat{\omega}^\dag$) acting on $\FF_0$,  where it consists of the set of decomposable elements in  $\FF^{(0)}_0$.
The finite dimensional criteria for decomposability can be extended to this context. The  Grassmannian $\Gr_{\HH_+} (\HH)$ is an orbit  space
 under the general linear group $\Gl(\HH)$ restricted (as in \cite{SW}), such that the orthogonal projection maps from elements
 of $\Gr_{\HH_+} (\HH)$ to ${\HH_+} $ are Fredholm, and of index zero. The corresponding Fock subspace $\FF_0\ss \FF$ is
 such that a generic decomposable element projects to a non zero multiple of the vacuum, which is the decomposable element
 corresponding to  ${\HH_+} $. The same then holds over the whole Fock space.

 Turning to criteria of decomposability, in finite dimensions decomposable elements
 $\phi \in\Lambda^k(\Cb^N)$ are those whose annihilators under the exterior product
\be
Ann(\phi) := \{\alpha\in  \Cb^N\ | \ \alpha \wedge \phi  = 0\}
\ee
 have maximal dimension $k$. In theinfinite dimensional setting, we have virtual dimensions for subspaces,
 given by the Fredholm index of  the  projections onto ${\HH_+} $ along $\HH_-$,  so we can define an
 element $|\phi\rangle \in \FF^{(0)}$ to be {\it virtually decomposable}  if its annihilator
 \be
Ann(|\phi\rangle := \{ v \in \HH \ | \  \hat{\Gamma}_v |\phi \rangle = 0\}
\ee
 is in $\Gr^0_{\HH_+} (\HH)$.

 We now truncate to finite dimensions. For $N\in \Nb^+$, let $\HH_N\ss \HH$  denote the $2N$ dimensional subspace
 spanned by $\{e_{-N }, e_{-N+1}, ..., e_{N-1}\}$. We then have the decomposition
 \be
 \HH = \HH_{N+} \oplus \HH_N \oplus \HH_{N-} \ee
 where $\HH_{N+}$ is the span of $\{e_{-N-i} \}_{i\in \Nb^+}$, and $\HH_{N-}$ is the span of $\{e_{N + i }\}_{i\in \Nb}$.
 Thus $\HH_N$  of eq.~(\ref{HH_N_def}) may be identified with $\HH_N$. Let
  \be
  \pi_N:  \HH \rightarrow \HH_{N+} \oplus \HH_N
  \ee
   denote projection to $\HH_{N+} \oplus \HH_N$  along $\HH_{N-} $, and
  \be
  \hat{\pi}_N: \FF_0 \ra \FF_{(0,N)}
  \label{proj_FF0_FF0N}
  \ee
 the corresponding projection from   $\FF_0$ to the subspace $\FF_{(0,N)}\ss \FF_0$
    spanned by those basis elements that have no factors in $\{e_{N + i }\}_{i\in \Nb}$.
Decomposable elements of this space are the Pl\"ucker image of subspaces
of codimension $N$  in $\HH_{N+} \oplus \HH_N$.

Consider the vacuum element
  \be
 | \hspace{-2 pt} -\hspace{-2 pt} N\rangle  = e_{-N-1 } \swedge e_{-N-2} \swedge \cdots
  \ee
  in the fermionic charge sector $\FF_{-N}$.
  In analogy with finite dimensions, the inner product map
  \be
   i_{| \hspace{-2 pt} -\hspace{-2 pt} N\rangle}: \FF_{(0,N)} \ra \Lambda^N(\HH_N)
   \ee
   is defined on basis elements by
   \be
 i_{| \hspace{-2 pt} -\hspace{-2 pt} N\rangle}( e_{l_1} \swedge e_{l_2} \swedge \cdots  ) =
\begin{cases}  e_{l_1} \swedge e_{l_2} \swedge \cdots  \swedge e_{l_N}  \ \text{if } l_j= -j , \ \forall \  j> N \cr
    0 \quad \text{otherwise}.
\end{cases}
  \ee
  Now define
   \be
  \phi_N^N := i_{| \hspace{-2 pt} -\hspace{-2 pt} N\rangle} (\hat{\pi}_N(|\phi \rangle) \in \Lambda^N(\HH_N) \ss \Lambda(\HH_N),
    \label{phi_N_N_def}
    \ee
    and
    \be
     |\phi_N\rangle :=  \hat{\Gamma}_{\phi^N_N }  | \hspace{-2 pt} -\hspace{-2 pt} N\rangle \in \FF_0,
  \ee
  where the Grassmann algebra $\Lambda(\HH_N)$ is identified with the finite dimensional subalgebra of
  the fermionic Clifford algebra $\Cl(\HH\oplus\HH^*, Q)$ generated by $\{\psi_{-N}, \cdots, \psi_{N-1}\}$.
  We then have
  \begin{lemma}
  The elements $|\phi_N\rangle$ converge to $|\phi \rangle$ as $N \ra \infty$.
  \end{lemma}

  Let $|\phi\rangle$ be  virtually decomposable. The annihilator  $Ann(|\phi\rangle)$ of $ |\phi\rangle$ is in $\Gr^0_{\HH_+}(\HH)$,
 so it has a Fredholm projection onto $\HH_+$, and a small (compact, or Hilbert Schmidt \cite{SW}) projection to $\HH_-$.
  Since its virtual dimension is zero, this means that its intersection
\be
Ann(|\phi\rangle)\cap (\HH_{N-} \oplus \HH_N)
\ee
 has dimension $N$, for large $N$, and the  projection
\be
Ann(|\phi\rangle)^N_N \subset \HH_N
\ee
of $Ann(|\phi\rangle)\cap (\HH_{N-} \oplus \HH_N)$ to $\HH_N$  has dimension $N$ for large $N$.

  \begin{lemma}
 Let $|\phi\rangle$ be virtually decomposable. Then, for large $N$,  $Ann(|\phi \rangle)^N_N$ is the annihilator of $\phi^N_N$,
 which  implies that $\phi^N_N$ is decomposable in  $\Lambda^N(\HH_N)$ and $|\phi_N\rangle$ is decomposable in $\FF_0$.
  \end{lemma}
The  annihilator $Ann(|\phi_N\rangle)$ of $|\phi_N\rangle$ is obtained  by adding to $Ann(\phi^N_N)$  the vectors $ e_{- N-1}, e_{ -N-2 }, e_{-N-3} ,  \dots$,
to obtain an infinite dimensional space of virtual dimension $0$.
\begin{lemma} 
The annihilators $Ann(|\phi_N\rangle)$ converge to $Ann(|\phi\rangle)$.
 \end{lemma}

Let $I$ be an infinite multi-index $I_1<I_2<\cdots$ such that for $j$ beyond a certain $N_I$, $I_j = j+\ell_0$ for a fixed $\ell_0$. 
The Fredholm property tells us that there is a coordinate plane $w_I$ corresponding to such an  $I$,  spanned by vectors $e_{-I_j}$,   
such that the projection of $|\phi\rangle$ to $ e_{-I_1} \swedge e_{-I_2} \swedge e_{-I_3} \cdots$  is nonzero. 
Going now to our decomposable $|\phi_N\rangle$,  for  $N>N_0$, the $|\phi_N\rangle$,'s
 also map non-trivially, and since  they  correspond to subspaces  $w_N$, they have bases $w_{N,1}, w_{N,2}, ....$,  
 of the form
\be
w_{N, j} = e_{-I_j} + \sum_{i\notin I} a_{N, i} e_{-i}.
\ee
These $w_{N, j}$ converge individually as $N\ra \infty$, for each $j$, and the limits $w_{\infty, j}$
 then give a decomposition 
 \be
 |\phi\rangle = w_{\infty, 1} \swedge w_{\infty, 2}  \swedge \cdots.
 \label{decomp_phi_infty}
 \ee
 \begin{proposition}
Suppose that for $N>N_0$, the elements $\phi_N^N$ are nonzero and decomposable; then $|\phi\rangle$ is also.
 If $|\phi\rangle$ is virtually decomposable, then it is decomposable.
 \end{proposition}

 An element $|\phi\rangle\in \FF^{(0)}_0$  will  be the Pl\"ucker image of an element of the Lagrangian Grassmannian if and
 only if  it is virtually decomposable, since the Lagrangian condition is guaranteed by its belonging to $\FF^{(0)}_0$. Furthermore, 
 the finite dimensional elements $\phi^N_N$ must also be isotropic with respect to the (finite-dimensional) symplectic form $\omega_N$.
 Thus,
 \begin{proposition}
The element $|\phi\rangle$ corresponds to a Lagrangian subspace if and only if the elements $\phi^N_N$ correspond to
 Lagrangian subspaces in finite dimensions for all $N>N_0$.
\end{proposition}
Recall the definition (\ref{p_B_I_proj_j}) of the projection map
\be
p^j_{(B,I)}: \Lambda^j(\HH_N) \ra  \Lambda^j(\HH_N) .
\ee
By Proposition (\ref{prop:plucker-hyper}), a generic $\phi_N^N$ is decomposable and represents a Lagrangian subspace if and only if,
 for all multi-indices $(A,I)$ of cardinality $N-3$ and complementary $(B,I)$ (as defined in (\ref{marked_index})), 
 the elements  $p^3_{(B,I)}(i_{f_{(A,I)}}(\phi_N^N) )= i_{f_{(A,I)}}(p^N_{(B,I)}(\phi_N^N))$ represent a Lagrangian $3$-space in $6$ dimensions. 
 Since in this section we have  different conventions for the numbering of elements of the basis,
  we redefine, for a fixed $N$, multi-indices $(A,I)_N, (B,I)_N$,  where
  \begin{subequations}
\bea
 I &\& = \{I_1,..,I_{N-3}\} \ss \{1,2,...,N\} \text{ is a subset of cardinality }  N-3 , \\
A  &\& \  \text{ associates to each } \ I_j\in I\ \text{ an integer }  A(I_j) \ \text{which is either } \ I_j-1, \ \text{ or} \  -I_j, \cr
&\& \\
 B &\& \ \text{associates to each}  \ I_j\in I \ \text{ an integer }   B(I_j)  \text{ which is complementary to } A(I_j),\cr
  &\& \text{so that either } \ B(I_j) = I_j-1,  \text{ if}\  A(I_j)= -I_j, \text{ or }  \ B(I_j) =-I_j  \text{ if }  A(I_j)= I_j-1. \cr
 &\&
\eea
\end{subequations}
We can  define multivectors $f_{(A,I)_N}$  on the spaces $\HH_N$, as wedge products of the elements $e_{A(i_j)}$, ordered so that the $A(i_j)$ are increasing.
 This gives corresponding contractions $i_{f_{(A,I)_N}}$. As above, we can project out  all the basis elements $e_{B(i_j)}$, and obtain a projection $p_{(B,I)_N}$.
 And, as above, $\{p^3_{(B,I)}i_{f_{(A,I)}}(\phi_N^N)\}$ give us $3$-vectors in $6$-space that correspond to Lagrangian subspaces for $\phi^N_N$
 to correspond to one.

There is a natural extension $(A^+,I^+)_{N+1}, (B^+,I^+)_{N+1} $ of $(A,I)_N, (B,I)_N$ from $N$ to $N+1$, given by:
\begin{subequations}
\bea
 I^+ &\&=  \{I_1,..,I_{N-3}, N+1\} , \\
 A^+(I_j) &\&= A(I_j), j\leq N-3, \ \text{ and}\  A^+(N+1)= -N-1, \\
 B^+(I_j) &\&= B(I_j), j\leq N-3, \ \text{ and}\  B^+(N+1)= N.
\eea
\end{subequations}
We  can stabilise to infinite dimensions, and  define semi-infinite multi-indices $(A,I)_\infty, (B,I)_\infty$ as follows.
\begin{subequations}
\bea
 I &\& =\{I_1,I_2,...\} \text{ is a subset  of the positive integers, omitting only $3$ integers }, \cr
 &\& \\
 A &\& \text{  associates to each } I_j \text{ an integer } A(I_j) \text{ which is either }  A(I_j)= I_j-1 \text{ or} -I_j.\cr
 &\& \text{ For } j \text{ greater than}   \text{ some } j_0, A(I_j) = -I_j.\\
 B &\& \text{  associates to each } I_j \text{ the integer } B(I_j) \text{\ ``complementary" to } A(I_j),\cr
 &\& \text{ that is }  B(I_j)= I_j-1  \text{ if } A(I_j)= -I_j, \text{ or }B(I_j) =\ -I_j \text{ if } A(I_j)= I_j-1.\cr
 &\& \text{ For } j \text{ greater than some } j_0, \ B(I_j) =  I_j-1.
\eea
\end{subequations}
Fix $(A, I)_\infty, (B, I)_\infty$, and let  $\HH_{(A,B,I)^c}$ be the six dimensional space defined in the same way 
as the finite dimensional case;  i.e.,  spanned by the basis elements indexed by integers not in $(A,I)$ or $(B,I)$. 
Now pick an $N>j_0$, and set $(A, I_{\leq N})_\infty, (B, I_{\leq N})_\infty$ to be the multi-index formed  
by the  indices of $(A,I)_\infty, (B,I)_\infty$ with $I_j\leq N$. Since $N>j_0$, they are of cardinality $N-3$, 
and are formed by the removal of the infinite tails  $A(I_j) =-I_j$, $ B(I_j) = I_j-1$ for $I_j>N$. 
From this we have  contractions $i_{f_{(A,I)_\infty}}$ and   projections $p^3_{(B,I)_\infty}$ given by 
\bea
 i_{f_{(A,I)_\infty}} =&i_{f_{(A,I_{\leq N})_\infty}}\circ i_{| \hspace{-2 pt} -\hspace{-2 pt} N\rangle} \\
 p^3_{(B,I)_\infty} =&p^3_{(B,I_{\leq N})_\infty }\circ \hat{\pi}_N
 \eea
 The composition  
 \be
  i_{f_{(A,I)_\infty}} p^3_{(B,I)_\infty} = i_{f_{(A,I_{\leq N})_\infty}}\circ p^3_{(B,I_{\leq N})_\infty}\circ  i_{| \hspace{-2 pt} -\hspace{-2 pt} N\rangle}\circ \hat{\pi}_N
 \ee
  then maps us to $\Lambda^3(\HH_{(A,B,I)^c})$, as in the finite dimensional case, passing through $\HH_N$ as an intermediary step. 
  The result is invariant under the stabilization from $N$ to $N+1$.
  Thus, the elements in dimensions $(3,6)$ that we must test, for $\phi^N_N$ to correspond to a Lagrangian plane, 
  are obtainable directly from $|\phi\rangle$ as $ i_{f_{(A,I)_\infty}}p^3_{(B,I)_\infty}(|\phi\rangle)$, 
  and belong to $\Lambda^3(\HH_{(A,B,I)^c})$.  
\begin{proposition} 
\label{prop:hypedet_inf_red}
A generic $|\phi\rangle \in \FF_0$ corresponds to an element of the Lagrangian Grassmannian if all of its $(3,6)$ dimensional reductions
$ p^3_{(B,I)_\infty}i_{f_{(A,I)_\infty}}(|\phi\rangle)$ correspond to Lagrangian planes.
This in turn is equivalent, for generic elements,  to $ p^3_{(B,I)_\infty}i_{f_{(A,I)_\infty}}(|\phi\rangle)$ satisfying the hyperdeterminantal relations (\ref{hyper_det_rels_inf}).
\end{proposition}


 \subsubsection{Parametric families of  hyperdeterminantal relations in terms of $\tau^{KP}_{w^0}$}
 \label{param_hyperdet_tau}

 \quad \ Choose three arbitrary parameters $(x_1, x_2, x_3)$, such that $x_i+x_j\ne 0$ for any distinct pair $i, j \in \{ 1,2,3\}$,
 and define the $3\times 3$ matrix valued function  $A(\tb', x_1, x_2, x_3)$  of the parameters $(x_1, x_2, x_3)$ and the
 odd KP flow parameters $\tb'=(t_1, 0, t_3, 0, \dots)$ with matrix elements
 \be
 A_{ij} (\tb', x_1, x_2, x_3) :=\frac{\tau_{w^0}^{KP}(\tb' + [x_i] - [-x_j])}{(x_i + x_j) \tau_{w^0}^{KP}(\tb')} , \quad i,j\in \{1,2,3\},
 \ee
 where $\tau_{w^0}^{KP}(\tau)$ is a KP $\tau$-function satisfying the condition (\ref{tau_KP_CKP_symm}) assuring that
 it generates solutions to the CKP hierarchy.
 It follows  that $A(\tb', x_1, x_2, x_3)$ is a symmetric matrix
 \be
 A(\tb', x_1, x_2, x_3)= A^T(\tb', x_1, x_2, x_3).
 \ee

 Define the following evaluations of $\tau^{KP}_{w^0}(\tb' )$
 \begin{subequations}
 \bea
 \sigma_{0}(\tb', x_1, x_2, x_3)&\&:= \tau^{KP}_{w^0}(\tb' ) , \\
 \sigma_{i} (\tb', x_1, x_2, x_3)&\&:=\frac{1}{2x_i}\tau^{KP}_{w^0}(\tb' +[x_i]-[-x_i]), \quad i=1,2,3\\
 \sigma_{0^*}(\tb', x_1, x_2, x_3) &\&:= \frac{\prod_{1\le i<j}^3 (x_i-x_j)^2}{\prod_{i,j= 1}^3(x_i+x_j)}\tau^{KP}_{w^0}(\tb' +\sum_{i=1}^3([x_i]-[-x_i])) , \\
 \sigma_{i^*}(\tb', x_1, x_2, x_3)  &\&:= \frac{(x_j-x_k)^2}{4x_j x_k (x_j+x_k)^2}\tau^{KP}_{w^0}(\tb' +[x_j] + [x_k]-[-x_j]- [-x_k]) , \cr
  &\&\cr
 \text{where } (i,j,k)  &\&\    \text{ is a} \text{ cyclic permutation of} \ (1,2,3).
 \eea
 \end{subequations}
 \begin{proposition}
 \label{tau_functional_hyperdet}
 \label{hyperdet_3_param_family}
 These satisfy the parametric family of hyperdeterminantal relations
 \bea
&\& \sigma_0^2\sigma_{0^*}^2  + \sigma_{1}^2\sigma_{1^*}^2  +\sigma_2^2\sigma_{2^*}^2  +\sigma_3^2\sigma_{3^*}^2
=2\sigma_0\sigma_{0^*} \sigma_{1}\sigma_{1^*} + 2\sigma_0\sigma_{0^*}\sigma_2\sigma_{2^*}+2 \sigma_0\sigma_{0^*} \sigma_{3}\sigma_{3^*} +
 \cr
 &\&  + 2 \sigma_{1}\sigma_{1^*} \sigma_{2}\sigma_{2^*}  + 2 \sigma_{1}\sigma_{1^*}  \sigma_3\sigma_{3^*}+ 2  \sigma_{2}\sigma_{2^*} \sigma_3\sigma_{3^*}
 -  4\sigma_{0^*}\sigma_1  \sigma_2\sigma_3  - 4\sigma_{0}\sigma_{1^*} \sigma_{2^*}\sigma_{3^*} \cr
 &\&
  \label{hyperdet_3_6_sigma}
\eea
 for all $(\tb', x_1, x_2, x_3)$.
\end{proposition}
 \begin{proof}
 Denote the  eight principal minors of $A(\tb', x_1, x_2, x_3)$,
 \begin{subequations}
 \bea
 \Sigma_0(\tb',x_1,x_2,x_3) &\&:= 1, \\
 &\&\cr
 \Sigma_i(\tb', x_1, x_2, x_3)&
\& := A_{ii} (\tb', x_1, x_2, x_3)  \quad i \in \{1,2,3\}, \\
&\&\cr
  \Sigma_{0^*}(\tb', x_1, x_2, x_3)  &\&:= \det\left(A(\tb', x_1, x_2, x_3)\right), \\
  &\&\cr
   \Sigma_{i*}(\tb', x_1, x_2, x_3)&\&:= \det\left(\begin{matrix}
  A_{jj} (\tb', x_1, x_2, x_3)  & A_{jk} (\tb', x_1, x_2, x_3) \\
   &\&\cr
   A_{kj} (\tb', x_1, x_2, x_3)  & A_{kk} (\tb', x_1, x_2, x_3)
   \end{matrix}\quad  \right), \\
   &\&\cr
 \text{where } (i,j,k)    \text{is a}  &\&\   \text{ cyclic  permutation of} \ (1,2,3).
 \nonumber
 \eea
 \end{subequations}
Since $A(\tb', x_1,x_2,x_3)$ is symmetric, these  satisfy the hyperdeterminantal relation
 \bea
&\& \Sigma_0^2\Sigma_{0^*}^2  + \Sigma_{1}^2\Sigma_{1^*}^2  +\Sigma_2^2\Sigma_{2^*}^2  +\Sigma_3^2\Sigma_{3^*}^2
=2\Sigma_0\Sigma_{0^*} \Sigma_{1}\Sigma_{1^*} + 2\Sigma_0\Sigma_{0^*}\Sigma_2\Sigma_{2^*}+2 \Sigma_0\Sigma_{0^*} \Sigma_{3}\Sigma_{3^*} +
 \cr
 &\&  + 2 \Sigma_{1}\Sigma_{1^*} \Sigma_{2}\Sigma_{2^*}  + 2 \Sigma_{1}\Sigma_{1^*}  \Sigma_3\Sigma_{3^*}+ 2  \Sigma_{2}\Sigma_{2^*} \Sigma_3\Sigma_{3^*}
 - 4\Sigma_{0^*}\Sigma_1  \Sigma_2\Sigma_3 - 4\Sigma_{0}\Sigma_{1^*} \Sigma_{2^*}\Sigma_{3^*}. \cr
 &\&
  \label{hyperdet_3_6_Sigma}
\eea

 Now recall the following consequence of the addition formula for KP $\tau$-functions (\cite{Shig},  and \cite{HB}, Chapt.~3, Prop. 3.10.4):
   \bea
    \label{addition}
\frac{\tau^{KP}  (\tb +\sum_{i=1}^k[x_i]-\sum_{i=1}^k[y_i])}{\tau^{KP}(\tb)} \frac{\prod_{i<j}(x_i-x_j)(y_j-y_i)}{\prod_{i,j= 1}^k(x_i-y_j)} &\&= \det\left(\frac{\tau^{KP} ( \tb + [x_i]-[y_j])}{(x_i-y_j)\tau^{KP}(\tb)}\right) _{1\leq i,j\leq k}.\cr
&\& \cr
&\&
\eea
Setting
\be
\tau^{KP}=\tau^{KP}_{w^0}, \quad \tb = \tb', \quad y_i := -x_i, \quad i=1,2,3,
\ee
and  choosing $k=0, 1, 2$ or $3$, this gives
\begin{subequations}
\bea
 \sigma_{0}(\tb', x_1, x_2, x_3)&\&=\tau^{KP}_{w^0}(\tb')\Sigma_{0}(\tb', x_1, x_2, x_3), \\
 \sigma_{i} (\tb', x_1, x_2, x_3)&\&= \tau^{KP}_{w^0}(\tb' )  \Sigma_{i}(\tb', x_1, x_2, x_3), \quad i=1,2,3\\
 \sigma_{0^*}(\tb', x_1, x_2, x_3) &\&= \tau^{KP}_{w^0}(\tb' ) \Sigma_{0^*}(\tb', x_1, x_2, x_3), \\
 \sigma_{i^*}(\tb', x_1, x_2, x_3)  &\&= \tau^{KP}_{w^0}(\tb' )  \Sigma_{i^*}(\tb', x_1, x_2, x_3) , \\
 \text{where } (i,j,k)  \   \text{is a} &\&\  \text{ cyclic  permutation of} \ (1,2,3).
 \nonumber
 \eea
 \end{subequations}
The hyperdeterminantal relation (\ref{hyperdet_3_6_Sigma}) may therefore be written  equivalently as (\ref{hyperdet_3_6_sigma}).
\end{proof}

In fact, there is no reason to limit the number of parameters to just $3$.
For any $\tau$-function $\tau^{KP}_{w^0}(\tb)$ of Lagrangian type, 
choose a set of $N$ parameters $\{x_i\}_{i =1, \dots, N}$ satisfying
\be
x_i + x_j \neq 0, \quad \forall\  i, j \in (1, \dots,  N)
\ee
 (where, in principle, we could allow $N \ra \infty$, provided suitable convergence conditions are satisfied),
and an arbitrary point $\tb'$ in the space of (odd) flow parameters. Then define the map 
\bea
\tau: \Zb^N &\& \ra \Cb \cr
\tau: \nb: &\& \mapsto \tau^{\nb}  := \tau^{KP}_{w^0}(\tb' + \sum_{i=1}^N n_i ([x_i]- [-x_i])) \\
 \nb&\&=(n_1, \dots, n_N) \in \Zb^N
 \nonumber
\eea
and, for each triple of integers $(i,j,k)$,  $1\le i < j < k\le N$,  the  eight quantities
\begin{subequations}
 \bea
 \sigma^\nb&\&:= \tau^\nb , 
 \label{sigma_tau_sub_a}\\
 \sigma^\nb_{a} &\&:=\frac{1}{2x_a}\tau^{(n_1, \dots,    n_a+1,  \dots, n_N)}, \quad a = i,j,k
 \label{sigma_tau_sub_b} \\
 \sigma^\nb_{ijk}&\&:= \frac{(x_i-x_j)^2 (x_j - x_k)^2(x_k-x_i)^2}{(x_i+x_j)^2 (x_j + x_k)^2(x_k +x_i)^2} 
\ \tau^{(n_1, \dots,  n_i+1,  \dots, n_j+1, \cdots n_k+1, \cdots n_N)},
\label{sigma_tau_sub_c} \\
 \sigma^\nb_{ab} &\&:= \frac{(x_a-x_b)^2}{4x_a x_b(x_a+x_b)^2}\tau^{(n_1, \dots,    n_a+1,  \dots, n_b+1,  \cdots n_N)} ,
\quad  (a,b) = (i,j), (j,k), (i,k). \cr
&\& 
\label{sigma_tau_sub_d}
   \eea
 \end{subequations}
 We then have:
 \begin{corollary}
   \label{tau_functional_N_param_hyperdet}
 For all $i<j<k$, the following $N$-parameter family of hyperdeterminantal relations hold:
 \label{hyperdet_N_param_family}
 \bea
&\& (\sigma^\nb \sigma^\nb_{ijk})^2  + (\sigma^\nb_i \sigma^\nb_{jk})^2 + (\sigma^\nb_j\sigma^\nb_{ki})^2  + (\sigma^\nb_k \sigma^\nb_{ij})^2 
=2\sigma^\nb\sigma^\nb_{ijk} \sigma^\nb_{i}\sigma^\nb_{jk} + 2\sigma^\nb\sigma^\nb_{ijk} \sigma^\nb_j\sigma^\nb_{ki} + 
2\sigma^\nb\sigma^\nb_{ijk} \sigma^\nb_k\sigma^\nb_{ij} + 
 \cr
 &\&  + 2 \sigma^\nb_{i}\sigma^\nb_{jk} \sigma^\nb_{j}\sigma^\nb_{ki}  + 2 \sigma^\nb_{j}\sigma^\nb_{ki} \sigma^\nb_{k}\sigma^\nb_{ij} 
 + 2 \sigma^\nb_{k}\sigma^\nb_{ij} \sigma^\nb_{j}\sigma^\nb_{ki} 
 -  4\sigma^\nb_{ijk}\sigma^\nb_i  \sigma^\nb_j\sigma^\nb_k  - 4\sigma^\nb\sigma^\nb_{ij} \sigma^\nb_{jk}\sigma^\nb_{ki} .\cr
 &\&
  \label{hyperdet_N_param_sigma}
\eea
 \end{corollary}
 The proof is the same as for Proposition \ref{tau_functional_hyperdet}, with the replacements
 \be
 (x_1, x_2, x_3) \ra (x_i, x_j, x_k), \quad
 \tb' \ra \tb' + \sum_{i=1}^N n_i ([x_i]- [-x_i]).
 \ee
 To obtain (\ref{hyper_det_rels}) from (\ref{hyperdet_N_param_sigma}), set $(i,j,k) = (j_1, j_2, j_3)$ and $\nb=\nb_J$, 
 the binary vector with $1$'s in positions $(J_1, \dots, J_r)$ and $0$'s elsewhere:
 \bea
 \LL_J := \sigma^{\nb_J}, \quad  \LL_{J, j_a} &\&:= \sigma^{\nb_J}_{j_a}, \quad   \LL_{J, j_a, j_b } := \sigma^{\nb_J}_{j_a, j_b},
 \quad  \LL_{J, j_1,j_2, j_3}  := \sigma^{\nb_J}_{(j_1, j_2, j_3)}, 
 \eea
 for $a, b \in\{1,2,3\}, \ a<b$.
Defining
\be
T_{n_i} (\tau^\nb):= \tau^{(n_1, \dots,    n_i+1,  \dots, n_N)}, \quad i\in \{1, \dots, N\}
\ee
and substituting eqs.~(\ref{sigma_tau_sub_a}) - (\ref{sigma_tau_sub_d})  into (\ref{hyperdet_N_param_sigma})
gives the form of the discrete CKP relations studied in \cite{FN}.


\subsection{Summary of results and further developments}
\label{summery_further_devel}

We have shown that any KP $\tau$-function $\tau_{w^0}$ satisfying the CKP reduction conditions 
(eq.~(\ref{inf_lagrangian_cond}), or eqs.~(\ref{null_fermi_omega}), (\ref{h_sigma_1_inv}) or (\ref{tau_KP_CKP_symm})),
corresponds to an element $w^0$ belonging to the subgrassmannian $\Gr^\LL_{\HH_+}(\HH, \omega) \ss \Gr_{\HH_+} (\HH)$
of Lagrangian subspaces of the symplectic Hilbert space $(\HH, \omega)$, acted on by the subgroup consisting of
those abelian flow group $\Gamma_+$ that preserve $\Gr^\LL_{\HH_+}(\HH, \omega)$ (i.e., the odd parameter flows
only). It was proved, as a consequence of the addition formulae for KP $\tau$-functions, that
any such $\tau$-function of CKP type, when evaluated at the finite or infinite lattice points
and normalized as in eqs.~(\ref{sigma_tau_sub_a}-\ref{sigma_tau_sub_d}), provides infinite parametric families 
of solutions to the hyperdeterminantal relations, depending on the choice of parameters $\{x_i\}_{i\in \Zb}$ and the origin
 $\gamma(\tb)w_0$ in the group orbit in which the lattice is embedded. These relations may be interpreted as defining
solutions of the discretized lattice form of the CKP hierarchy \cite{Sch, FN}. Moroever,  as noted 
 in Section \ref{sec:hexahedron} for the finite dimensional case, these may also be extended to a more general system, 
 by adding  further Pl\"ucker coordinates of the  Lagrangian Grassmannian element, besides those entering in the Lagrange map, so as to provide solutions 
to the hexahedron recurrence relations of Kenyon and Pemantle \cite{KePe1, KePe2}. 

   This suggests that, by making suitable lattice evaluations, corresponding both to symmetric partitions and some further
 ``almost'' symmetric ones, we may derive, for any CKP type $\tau$-function, infinite families of solutions, both
   of the discrete CKP hierarchy, and the hexahedron recursion relations. The detailed development of these results
   is done in a subsequent paper \cite{AHH},  in which the addition  formulae for KP $\tau$-functions are shown to 
   imply that such normalized lattice evaluations of $\tau$-functions of the continuous hierarchies, both KP and CKP, 
   provide infinite parametric families of solutions to the discrete ones.
   A similar result holds for the BKP hierarchy, in which the Pl\"ucker relations, which in the KP case are
   equivalent to the Hirota bilinear residue relations, are replaced by the corresponding Cartan relations 
   (cf. \cite{Ca}, Sec. 7.2 and Appendix E of \cite{HB}, and \cite{BHH}), which play a similar role in the embedding of maximal 
   isotropic Grassmannians with respect to a complex scalar product  into the projectivization of the Fock 
   space of neutral fermions (\cite{DJKM1, DJKM2}).

 \bigskip
 \bigskip
\noindent
\small{ {\it Acknowledgements.} The authors would like to thank M. Jimbo, R. Kenyon,  J. van de Leur, L. Oeding
and A. Zabrodin for helpful exchanges that contributed greatly to clarifying the results presented here.
This work was partially supported by the Natural Sciences and Engineering Research Council of Canada (NSERC).

 \bigskip
\noindent
\small{ {\it Data sharing.}
Data sharing is not applicable to this article since no new data were created or analyzed in this study.
\bigskip


 \newcommand{\arxiv}[1]{\href{http://arxiv.org/abs/#1}{arXiv:{#1}}}

\end{document}